\documentclass[letterpaper]{article}
\usepackage{aaai}
\usepackage{times}
\usepackage{helvet}
\usepackage{courier}

\usepackage{float}
\usepackage{graphicx}
\usepackage{comment}
\usepackage{color}
\usepackage{amsthm, amsmath, amssymb}

\newtheorem{theorem}{Theorem}
\newtheorem{corollary}{Corollary}

\floatstyle{ruled}
\newfloat{algorithm}{tbp}{loa}
\floatname{algorithm}{Algorithm}

\usepackage{algorithmic}
\usepackage{subfig}
\makeatother
\title{Methods to Determine Node Centrality and Clustering \\in Graphs with Uncertain Structure}

\author{Joseph J. Pfeiffer, III \\ 
Department of Computer Science \\
Purdue University \\
West Lafayette, IN 47907 \\
\textit{jpfeiffer@purdue.edu} 
\And
Jennifer Neville \\
Departments of Computer Science and Statistics \\
Purdue University \\
West Lafayette, IN 47909\\
\textit{neville@cs.purdue.edu}
    }

\begin{document}
\setlength{\baselineskip}{11pt}

       \maketitle

\begin{abstract}
Much of the past work in network analysis has focused on analyzing discrete graphs, where binary edges represent the ``presence'' or ``absence'' of a relationship. 
Since traditional network measures (e.g., betweenness centrality) utilize a discrete link structure, complex systems must be transformed to this representation in order to investigate network properties.  
However, in many domains there may be \emph{uncertainty} about the relationship structure and any uncertainty information would be lost in translation to a discrete representation. 
Uncertainty may arise in domains where there is moderating link information that cannot be easily observed, i.e., links become inactive over time but may not be dropped
 or observed links may not always corresponds to a valid relationship. 
In order to represent and reason with these types of uncertainty, we move beyond the discrete graph framework and develop social network measures based on a {\em probabilistic} graph representation. More specifically, we develop measures of path length, betweenness centrality, and clustering coefficient---one set based on sampling and one based on probabilistic paths.  We evaluate our methods on three real-world networks from Enron, Facebook, and DBLP,  showing that our proposed methods more accurately capture salient effects without being susceptible to local noise, and that the resulting analysis produces a better understanding of the graph structure and the uncertainty resulting from its change over time.    
\end{abstract}


\section{Introduction}

Much of the past work in network analysis has focused on analyzing discrete graphs, where entities are represented as
nodes and binary edges represent the ``presence'' or ``absence'' of a
relationship between entities. Complex systems of relationships are first transformed to a discrete graph representation (e.g., a friendship graph) and then the connectivity properties of these graphs are used to investigate and understand the characteristics of the system.  
For example, network measures such as the average shortest path length and clustering coefficient have been used to explore the properties of biological and information networks \cite{clusteringcoefficient,Leskovec05graphsover}, while measures such as centrality have been used for determining the most important and/or influential people in social networks \cite{freeman77centrality,Bonacich:87}.  

The main limitation of measures defined for a discrete representation is that they cannot
easily be applied to represent and reason about \emph{uncertainty} in the link structure. Link uncertainty may arise in domains where graphs evolve over time, as links observed at a earlier time may no longer be present or active at the the time of analysis. For example in
online social networks, users articulate ``friendships" with other users
and these links often persist over time, regardless of whether the friendship is maintained. 
This can result in uncertainty about whether an observed friendship link is still \emph{active} at some later point in time.  In addition, there may be uncertainty with respect to the {\em strength} of the articulated relationships~\cite{relstrength},
which can result in uncertainty about whether an observed relationship will be used to transmit information and/or influence.
Furthermore, there are other network domains (e.g., gene/protein networks) where relationships can only be indirectly observed so there is uncertainty about whether an observed edge (e.g., protein interaction) actually indicates the
presence of a valid relationship.

In this work, we formulate a probabilistic graph representation to analyze domains with these types of uncertainty and develop analogues for three standard discrete graph measures---average shortest path length, 
betweenness centrality, and clustering coefficient---in the probabilistic setting.  Specifically, we
use probabilities on graph edges to represent link uncertainty and consider the {\em distribution} of possible (discrete) graphs that they define, then we develop measures that consider the properties of the graph population defined by this distribution.  

Our first set of measures compute {\em expected} values over the distribution of graphs, sampling a set of discrete graphs from this distribution in order to efficiently approximate the path length, centrality, and clustering measures. We then develop a second set of measures that can be directly computed from the probabilities, which removes the need for graph sampling. The second approach also affords us the opportunity to consider more than just shortest paths in the network. We note that previous focus on shortest paths is due in part to an implicit belief that short paths are more likely to result in successful transfer of information and/or influence between two nodes. This has led other  
works to generalize shortest paths to the probabilistic domain for their own purposes 
\cite{sampleprobpaths}. 
However, in a probabilistic framework we can also directly compute the likelihood of a path and consider the {\em most probable} paths, which are likely to facilitate information flow in the network.  

With probabilistic paths, we also introduce a {\em prior} to incorporate the  belief that the probability of successful information transfer is a function of path length---since the existence of a relationship does not necessarily mean that information/influence will be passed across the edge. 
This formulation, which models the likelihood of information spread throughout the graph, is consistent with the finding in
\cite{Onnela2007TieStrength}, which identified that constricting and relaxing
the flow along the edges in the network was necessary to model the true
patterns of information diffusion in an evolving communication graph. 


We evaluate our measures on three real world networks: Enron email, Facebook
micro communications, and DBLP coauthorships. In these datasets, the network transactions are each associated with timestamps (e.g., email date). Thus we are able
to compute the local (node-level) and aggregate (graph-level) measures at multiple time steps, where at each time step $t$ we consider the network information available up to and including $t$. We compare against two different approaches that use the discrete representation: an
\emph{aggregate} approach, which unions all previous transactions (up to $t$) into a discrete graph, and a \emph{slice} approach, where only transactions from a small window (i.e., $[t-\delta, t]$) are included in the discrete representation.  For our methods, we estimate edge probabilities from the transactions observed up to $t$, weighting each transaction with an exponential decay function. Our analysis shows that our proposed methods more accurately capture the
salient changes in graph structure compared to the discrete methods without being
susceptible to local, temporal noise. Thus the resulting analysis produces a
better understanding of the graph structure and its change over time.


\section{Related Work}
The notion of probabilistic
graphs have been studied previously, notably by \cite{frankshortestpaths}, \cite{probtraffic}
and \cite{sampleprobpaths}. \cite{frankshortestpaths} showed how
for graphs with probability distributions over the weights for
each edge, Monte Carlo methods can be used to sample to determine the shortest
path probabilities between the edges. \cite{probtraffic} then extends this to
find the shortest weighted paths most likely to complete within a certain time constraint (e.g., the shortest distance across town in under half an hour). 
In \cite{sampleprobpaths}, the most probable shortest paths are used to estimate the $k$-nearest neighbors in the graph for a particular node.  Although
\cite{sampleprobpaths} draws sample graphs based on {\em likelihood} (i.e., sampling each edge according to its probability),
in their estimate of the shortest path distribution they weight each sample graph based on its probability, which is incorrect unless the samples are drawn uniformly at random from the distribution. In this work, we sample in the same manner as \cite{sampleprobpaths}, but weight each sample uniformly in our expectation calculations---since, when the graphs are drawn from the distribution based on their likelihood, the graphs with higher likelihood are more likely to be sampled.  
 


There has also been some recent work that has developed measures for time-evolving graphs, e.g., to identify
the most central nodes throughout time \cite{tang-tempshortestpaths} and identify the edges that maximize  communication over time \cite{vectorclocks}.  However, these 
works fail to account for the uncertainty in both the link structure and the the communication across links (as users are unlikely to propagate
all information across a single edge).  Our use of a probabilistic graph framework and transmission prior address these two cases of uncertainty.

\section{Sampling Probabilistic Graphs}

Let $G=\left\langle V,E\right\rangle$, be a graph 
where $V$ is a collection of nodes and $E \in V \times V$ is the set of edges, or
relationships, between the nodes. In order to represent and reason about relationship uncertainty, we associate each edge $e_{ij}$ (which connects node $v_{i}$ and $v_{j}$) with a probability $P(e_{ij})$.
Then we can define $\mathcal{G}$ to be a distribution of discrete, unweighted graphs. Assuming independence among edges, the probability of a graph $G \in \mathcal{G}$ is: $P(G) = \prod_{e_{ij} \in E} P(e_{ij}) \prod_{e_{ij} \notin E} \left[1 - P(e_{ij}) \right]$.
Since we have assumed edge independence, we can sample a graph $G_S$ from $\mathcal{G}$ by sampling edges independently according to their probabilities $P(e_{ij})$. Based on this, we can develop methods to compute the \emph{expected} shortest path lengths, betweenness centrality rankings, and clustering coefficients using sampling.



\vspace{-2.mm}
\paragraph{Probabilistic Average Shortest Path Length}
Let $\rho_{ij}=\{v_{k_1}, v_{k_2}, ..., v_{k_q}\}$ refer to a \emph{path} of $q$ vertices connecting two vertices $v_i$ and $v_j$, i.e., $v_{k_1}=v_i$ and $v_{k_q}=v_j$, and from each vertex to the next there exists an edge: $e_{k_i k_{i+1}} \in E$ for $i=[1, q-1]$. Let $V(\rho_{ij})$
and $E(\rho_{ij})$ refer to the set of vertices and edges respectively, in the path and let $|\rho_{ij}| = |E(\rho_{ij})|$ refer to the {\em length} of the path.
Assuming connected graphs, for every unweighted graph $G\!\!=\!\!\left\langle V,E\right\rangle \in \mathcal{G}$ there exists a shortest path $\rho_{ij}^{min}$ between every pair of nodes $v_i, v_j \in V$.  Letting $\mbox{SP}_{ij} = |\rho_{ij}^{min}|$, we can then define the average shortest path length in $G$ as:
$\overline{\mbox{SP}}(G) = \frac{1}{|V|\cdot(|V|-1)}{\sum_{i \in V} \sum_{j \in V; j\neq i} \mbox{SP}_{ij}}$.

Now, when there is uncertainty about the edges in G, we can compute the {\em expected} average shortest path length by considering the distribution of graphs $\mathcal{G}$.
For any reasonable sized graph, the distribution $\mathcal{G}$ will be intractable
to enumerate explicitly, so instead we sample from $\mathcal{G}$ to approximate the expected value. More specifically, we sample a graph $G_s$ by sampling edges uniformly at random according to their edge probabilities $P(e_{ij})$. Each graph that we sample in this manner has equal likelihood, thus we can draw $m$ sample graphs $G_S=\{G_1, ..., G_m\}$ and calculate the expected shortest path length with the following:
\begin{equation}
\mathbb{E}_{\mathcal{G}}\!\left[\: \overline{\mbox{SP}} \: \right] =\sum_{G\in\mathcal{G}}  \overline{\mbox{SP}}(G) \cdot P(G) \simeq\frac{1}{m}\sum_{m}  \overline{\mbox{SP}}(G_m) 
\end{equation}

Since the sampled graphs are unweighted, it takes $O\left(\left|V\right|\left|E\right|\right)$ time to 
compute $\overline{\mbox{SP}}$ for each sample~\cite{Brandes01FasterBC}. This results in an overall cost of $O\left(m\cdot\left|V\right|\left|E\right|\right)$ to compute $\mathbb{E}_{\mathcal{G}}\!\left[\: \overline{\mbox{SP}} \: \right]$.

\vspace{-2.mm}
\paragraph{Sampled Centrality}
Betweenness centrality for a node $v_i$ is defined to be the number of shortest paths between other pairs of nodes which pass through  $v_i$: 
$BC_i = | \{ \rho_{jk}^{min} \in G :  v_i \in V(\rho_{jk}) \;  \wedge \; i \neq j,k\} |$.
Vertices that contribute to the existence of many shortest paths will have a higher BC score than 
other nodes that contribute to fewer shortest paths, thus BC is used a measure of importance or centrality in the network. It is difficult to directly compare BC values across graphs since the number of shortest paths varies with graph size and connectivity. Thus, typically analysis focuses on {\em betweenness centrality rankings} (BCR), where the nodes are ranked in descending order of their BC scores and the node with the highest BC score is given a BCR of 1. 

As discussed above, we can compute the shortest paths for each unweighted graph $G\in\mathcal{G}$, then we can also compute the BCR values for each unweighted graph $G\in\mathcal{G}$. We denote $\mbox{BCR}_{i}(G)$ as the betweenness centrality ranking for node $v_{i}$ in $G$. Then we can approximate the expected BCR for each node by sampling a set of $m$ graphs from $\mathcal{G}$:
\begin{equation}
\mathbb{E}_{\mathcal{G}}\!\left[ \mbox{BCR}_i \right] \simeq\frac{1}{m}\sum_{m}  \mbox{BCR}_i(G_m) 
\end{equation}
Again, since the sampled graphs are unweighted, it takes $O\left(\left|V\right|\left|E\right|\right)$ time to 
compute the BCR for each sample~\cite{Brandes01FasterBC}, resulting in an overall cost of 
$O\left(m\cdot\left|V\right|\left|E\right|\right)$.

\vspace{-2.mm}
\paragraph{Sampled Clustering Coefficients}
Clustering coefficient is a measure of how the nodes in a graph cluster together~\cite{clusteringcoefficient}. 
For a node $v_i$ with $N_i\!=\!\{v_{j_1}, ..., v_{j_n}\}$ neighbors (e.g., $e_{i j_1}\!\!\! \in \!\!\! E$), its clustering coefficient is defined as $\mbox{CC}_i \!=\! \frac{1}{|N_i|(|N_i|-1)}\sum_{v_j \in N_i} \sum_{v_k \in N_i, k\neq j} \mathbb{I}_E(e_{jk})$, where $\mathbb{I}_E$ is an indicator function which returns 1 if $v_j$ is connected to $v_k$.  CC can be thought of as the fraction connected pairs of neighbors of $v_i$.
We denote $\mbox{CC}_{i}(G)$ as the clustering coefficient for node $v_{i}$ in graph $G$. Similar to paths, we can compute clustering coefficients for every graph $G \in \mathcal{G}$.  Thus we can approximate the expected CC for each node by sampling a set of $m$ graphs from $\mathcal{G}$:
\begin{equation}
\mathbb{E}_{\mathcal{G}}\!\left[ \mbox{CC}_i \right] \simeq\frac{1}{m}\sum_{m}  \mbox{CC}_i(G_m) 
\end{equation}
Under the assumption that the maximum degree in the graph can be bounded by a fixed constant (which is typical for sparse social networks), we can compute the clustering coefficient
for a single graph in $O(\left|V\right|)$ time (i.e., $O(1)$ for each node), which results in an overall cost of $O(m\cdot \left|V\right|)$.

\section{Probabilistic Path Length}
In the previous section, we discussed how to extend the discrete notions of shortest paths and centrality into a probabilistic graph framework via expected values, and we showed how to estimate approximate values using sampling.  While our sampling-based measures are valid and give informative results (see section 6 for details), they have two limitations which restrict their applicability.  

First, the effectiveness of the approximation depends on the number of samples from $\mathcal{G}$. We note that \cite{sampleprobpaths} used a Hoeffding Inequality to show that relatively few samples are needed to compute an accurate estimate of independent shortest
paths in probabilistic graphs.  However, since our the calculation of BCR is based on the joint occurrence of shortest paths in the graph, this bound will not hold for our measures.   

Second, since the expectation is over possible worlds (i.e., $G \in \mathcal{G}$), the focus on shortest paths may no longer be the best way to capture node {\em importance}. We note that in the discrete framework, where all edges are equally likely, the use of shortest paths as a proxy for importance implies a prior belief that shorter paths are more likely to be used successfully to transfer information and/or influence in the network. 
In domains with link uncertainty, the flow of information/influence will depend on both the {\em existence} of paths in the network and the {\em use} of those paths for communication/transmission. In a probabilistic framework, we have an opportunity to explicitly incorporate the latter, by encoding our prior beliefs about transmission likelihood into measures of node importance. Furthermore, although a probabilistic representation enables analysis of more than just shortest paths, as we note above, even to capture shortest paths the sampling methods described previously may need many samples to accurately estimate the joint existence of shortest paths. Thus, a measure that explicitly uses the edge probabilities to calculate most \emph{probable} paths may more accurately highlight nodes that serve to connect many parts of the network.  We discuss each of these issues more below.

\vspace{-4.mm}
\paragraph{Most Probable Paths}
To begin, we extend the notion of discrete paths to probabilistic paths in our framework. Specifically, we can calculate the probability of the existence of a path $\rho_{ij}$ as follows (again assuming edge independence):
$P(\rho_{ij})=\prod_{e_{uv}\in E(\rho_{ij})}P(e_{uv})$.
Using the path probabilities, we can now describe
the notion of the {\em most probable} path. Given two nodes $v_i,v_j$, the most probable path
path is simply the one with {\em maximum likelihood}: $ \rho_{ij}^{ML} = \arg\!\max P(\rho_{ij})$. 
We can compute the most likely paths 
in much the same way that shortest paths are computed on weighted discrete graphs, by applying Dijkstra's shortest path algorithm, but instead of expanding on the shortest path, we expand the most probable path. 
Thus, all most probable paths can be calculated in
$O\left(\left|V\right|\left|E\right|+\left|V\right|^{2}\log\left|V\right|\right)$.

\vspace{-4.mm}
\paragraph*{Transmission Prior} Previous focus on shortest paths for assessing centrality points to an implicit assumption that if an edge connects two nodes that it can be successfully used for transmission of information and/or influence in the network. Although there has been work both in maximizing the spread of information in a network through the use of central nodes \cite{boragatti-netflow,betweennesscentralityrandomwalks} and in the study of information propagation through the use of transmission probabilities \cite{transmissionprob}, there has been little prior work that has incorporated transmission probabilities into node centrality measures. 
Centrality measures based on random walks and eigenvectors~\cite{betweennesscentralityrandomwalks} implicitly penalize longer paths as they consider {\em all} paths between nodes in the network. 
However, in our framework we can incorporate transmission probabilities to penalize the probabilities of longer paths in the graph, in order to more accurately capture the role nodes play in the spread of information across multiple paths in the network.

Consider the case where there is one path of nine people where each edge has high probability of existence (e.g., 0.95) and another path of three people where the edge probabilities are all moderate (e.g., 0.70), both ending at node $v$.  Here, the longer path is more likely to exist than the shorter path, but in this example we are more interested in which path is used to transfer a virus to $v$.  
Even when an edge exists (i.e., the relationship is active), the virus will not be passed with certainty to the next node, thus the \emph{transmission probability} is independent of the edge probability. Moreover, when the transmission probability is less than 1, it is more likely that the virus will be transmitted across the shorter path, since the longer path presents more opportunities for the virus to be dropped.
This provides additional insight as to why shortest paths have always been considered important---there is generally a higher likelihood of transmission if it is passed through fewer nodes in the network.

To incorporate transmission likelihood into our probabilistic paths, we assign a probability $\beta$ of success for every
step in a particular path---corresponding to the probability that information
is transmitted across an edge and is received by the neighboring node. If we
denote $l$ to be the length of a path $\rho$, and $s$ to be the number of
successful transmissions along the path, we can use a binomial distribution to represent  the transmission
probability across $\rho$ with:
\[
\mbox{SBin}(s|\beta)=\mbox{Bin}(s=l|l,\beta)= \beta^{l}\]
Here SBin corresponds to the case where the transmission \emph{always} succeeds (i.e., across all edges in $\rho$).
Using this binomial distribution as a prior allows us to represent the expected probability of
information spread in an intuitive manner, giving us a parameter $\beta$
which we can adjust to fit our expectations for the information spread
in the graph. Note that setting $\beta=1$ is equivalent to the most probable paths discussed earlier.
The prior effectively {\em handicaps} longer paths through the graph. Although, there is a correlation between shortest (certain)
paths and handicapped (uncertain) paths, these formulations are \emph{not} equivalent, since the latter produces a different set of paths when the shortest paths have low probability of existence. 

\vspace{-2.mm}
\paragraph*{ML Handicapped Paths} Now that we have both the notion of a probabilistic path, and an
appropriate prior for modeling the probability of information spreading
along the edges in the path, we can formulate the \emph{maximum likelihood
handicapped path} between two nodes $v_{i}$ and $v_{j}$ to be:
\begin{equation}
\rho_{ij}^{MLH}=\arg\!\max_{\rho_{ij}}\left[P(\rho_{ij})\cdot\mbox{SBin}(\:|\rho_{ij}|\: |\:\beta )\right]\label{eq:ml}\end{equation}
To compute the most likely handicapped (MLH) paths, we follow the
same formulation as the most probable paths, keeping track of the path length and posterior at each point.
In the MLH formulation,  probable paths are weighted by likelihood of transmission, thus nodes that lie on paths that are highly likely and relatively short, will have a high BC ranking.
To calculate BCR ranking based on MLH paths, we can use a weighted betweenness centrality algorithm. 
Specifically, we modify Brandes' algorithm~\cite{Brandes01FasterBC}
to start with the path that has the lowest probability of occurrence
to be the one to backtrack from, enabling computation of the betweenness
centrality in $O\left(\left|V\right|\left|E\right|+\left|V\right|^{2}\log\left|V\right|\right)$.%

\subsection{\label{sec:Theory}Comparison with Discrete Graphs}

The formulation of MLH Paths has
inherent benefits, most notably with its direct connection to the
previously well-studied notions of shortest paths and betweenness
centrality in discrete graphs. In fact, we can view a discrete graph $G$ as being a special
case of probabilistic graph with edge probabilities:
\vspace{-4.mm}
\begin{equation}
P(e_{ij})=\begin{cases}
1 & \mbox{if an edge exists}\\
0 & \mbox{if the edge does not exist}\end{cases}\label{eq:probstatic}\end{equation}
We denote the distribution of graphs defined by these probabilities as $\mathcal{G}_1$. Note that the only graph in $\mathcal{G}_1$ with non-zero probability is $G$---since if an edge exists in a discrete graph,
then it exists with complete certainty, likewise, if
an edge is not present, we are certain it does not exist, thus $P(G)=1$.

\begin{theorem}
For every pair of nodes $v_i$ and $v_j$, the shortest path in the discrete graph ($\rho_{ij} \in G$) is equal to the most probable path discovered by the MLH algorithm ($\rho_{ij}^{MLH} \in \mathcal{G}_1$), for $0<\beta<1$.
\end{theorem}

\begin{proof}
In $\mathcal{G}_1$ every $P(e_{ij})$
is either $1$ or $0$, thus every case where $P(\rho_{ij})>0$ is precisely
$P(\rho_{ij})=1$.  If we choose the shortest path from the discrete
graph, it will have length $l^* = |\rho_{ij}|$, and the MLH probability for the same
path will be $\beta^{l^*}$.  Clearly, if a longer path were chosen by MLH,
its probability would be less than $\beta^{l^*}$, and we know that no shorter paths exist---since all paths shorter than $\rho_{ij}$ would involve an edge than did not exist in $G$ and thus would have probability 0 .
\end{proof}
%


\vspace{-3.mm}
\begin{corollary}
The betweenness centrality using
shortest paths on a discrete graph $G$ can be equivalently calculated with most probable handicapped paths over $\mathcal{G}_1$, where edge probabilities
are defined by Equation \ref{eq:probstatic}.
\end{corollary}

\begin{proof}
This follows directly from Thm 1. 
\end{proof}


\section{Probabilistic Clustering Coefficient}

We now outline a probabilistic measure of clustering coefficient that can be computed without the need for sampling.  
If we assume independence between edges, the probability of a 
triangle's existence is equal to the product of the probabilities of the three sides. The expected number of triangles is then 
the sum of the triangles probabilities that include a given node $v_{i}$. Denoting $\mbox{Tr}_{i}$ to 
be the expected triangles including $v_{i}$:
$\mathbb{E}_{\mathcal{G}} \left[\mbox{Tr}_{i} \right]=\!\!\sum_{v_{j},v_{k}\in N_i,v_{j}\neq v_{k}}\!\! \left[P\left(e_{ij}\right)\cdot P\left(e_{ki}\right)\cdot P\left(e_{jk}\right) \right]$.
Denoting $\mbox{Co}_{i}$ to be the expected
\emph{combinations} (i.e., coexisting pairs) of the neighbors of $v_{i}$, we then get:
$\mathbb{E}_{\mathcal{G}}\left[\mbox{Co}_{i}\right]=\!\!\sum_{v_{j},v_{k}\in N_i,v_{j}\neq v_{k}}\!\!\left[P\left(e_{ij}\right)\cdot P\left(e_{ki}\right)\right]$.
We can then define the probabilistic clustering coefficient to be the expectation of the ratio $\mbox{Tr}_{i}/\mbox{Co}_{i}$, and approximate it via a first order Taylor expansion~\cite{taylorcite}:

\begin{equation}
\mbox{CC}_{i} =\mathbb{E}_{\mathcal{G}}\left[\frac{\mbox{Tr}_{i}}{\mbox{Co}_{i}}\right] \approx \frac{{\displaystyle \mathbb{E}_{\mathcal{G}}\left[\mbox{Tr}_{i}\right]}}{{\displaystyle \mathbb{E}_{\mathcal{G}}\left[\mbox{Co}_{i}\right]}}
\end{equation}


Assuming again that the maximum degree in the graph can be bounded by a fixed constant, we can compute the probabilistic clustering coefficient in $O(\left|V\right|)$ time ($O(1)$ for each node). Additionally, the probabilistic
approximation to the clustering coefficient shares connections with the traditional clustering coefficients on discrete
graphs.



\begin{theorem}
The probabilistic clustering coefficients computed in $\mathcal{G}_1$, with probabilities defined by \ref{eq:probstatic} for a discrete graph $G$, are equal to the discrete clustering coefficients calculated on $G$.
\end{theorem}

\begin{proof}
Any triangle from $G$ has probability 1 in $\mathcal{G}_1$, 
while any non-triangle
in $G$ clearly has probability 0. The same is true for the combinations of pairs of neighbors. As such, the sums of the
numerators and denominators will be equal for both clustering coefficient. 
\end{proof}

\section{Experiments}
 \begin{figure*}[t]
\begin{centering}
\vspace{-4.mm}
\subfloat[]{\begin{centering}
\includegraphics[width=0.18\textwidth,height=2.6cm]{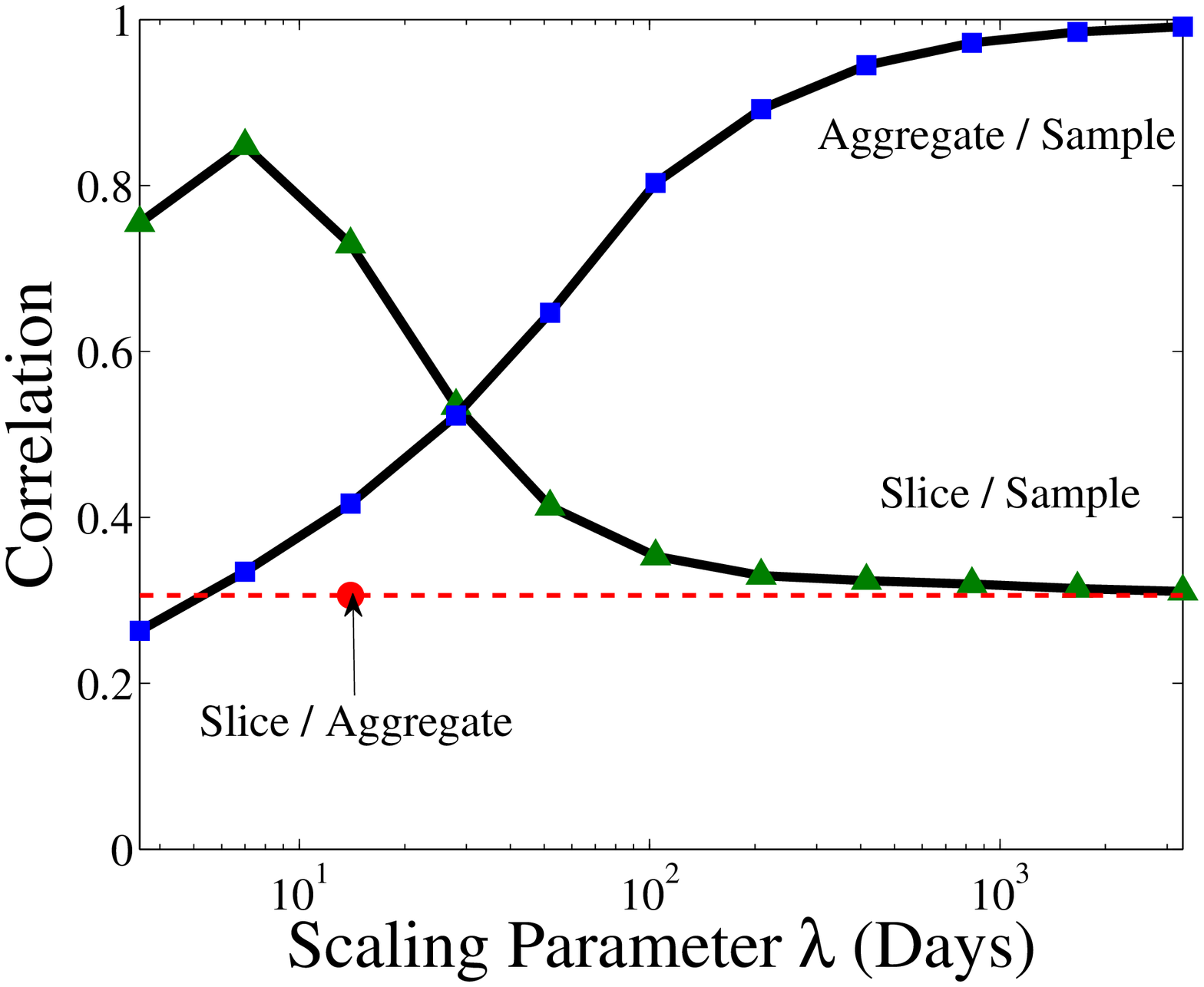}
\par\end{centering}}
\subfloat[]{\begin{centering}
\includegraphics[width=0.18\textwidth,height=2.6cm]{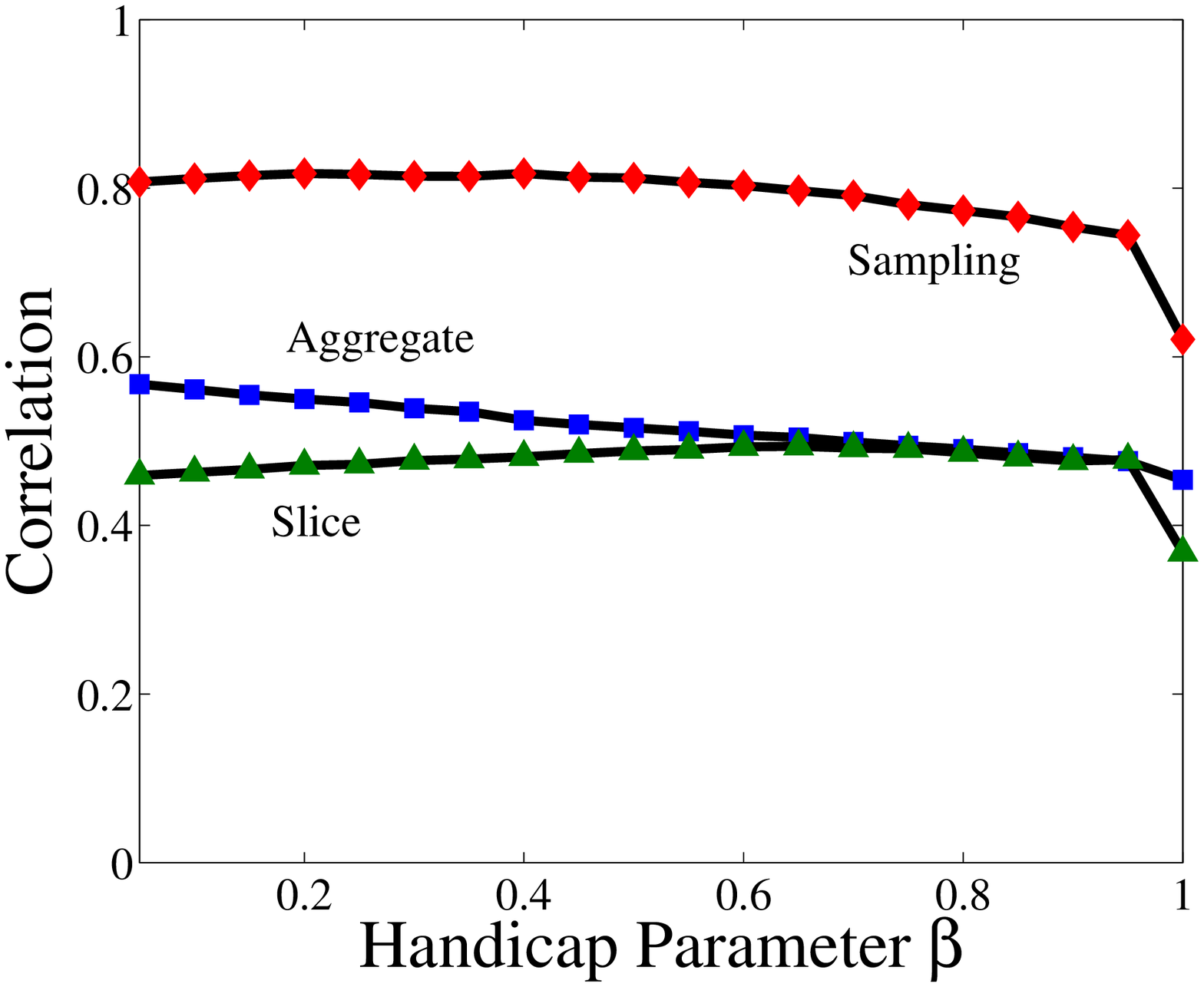}
\par\end{centering}}
\subfloat[]{\begin{centering}
\includegraphics[width=0.18\textwidth,height=2.6cm]{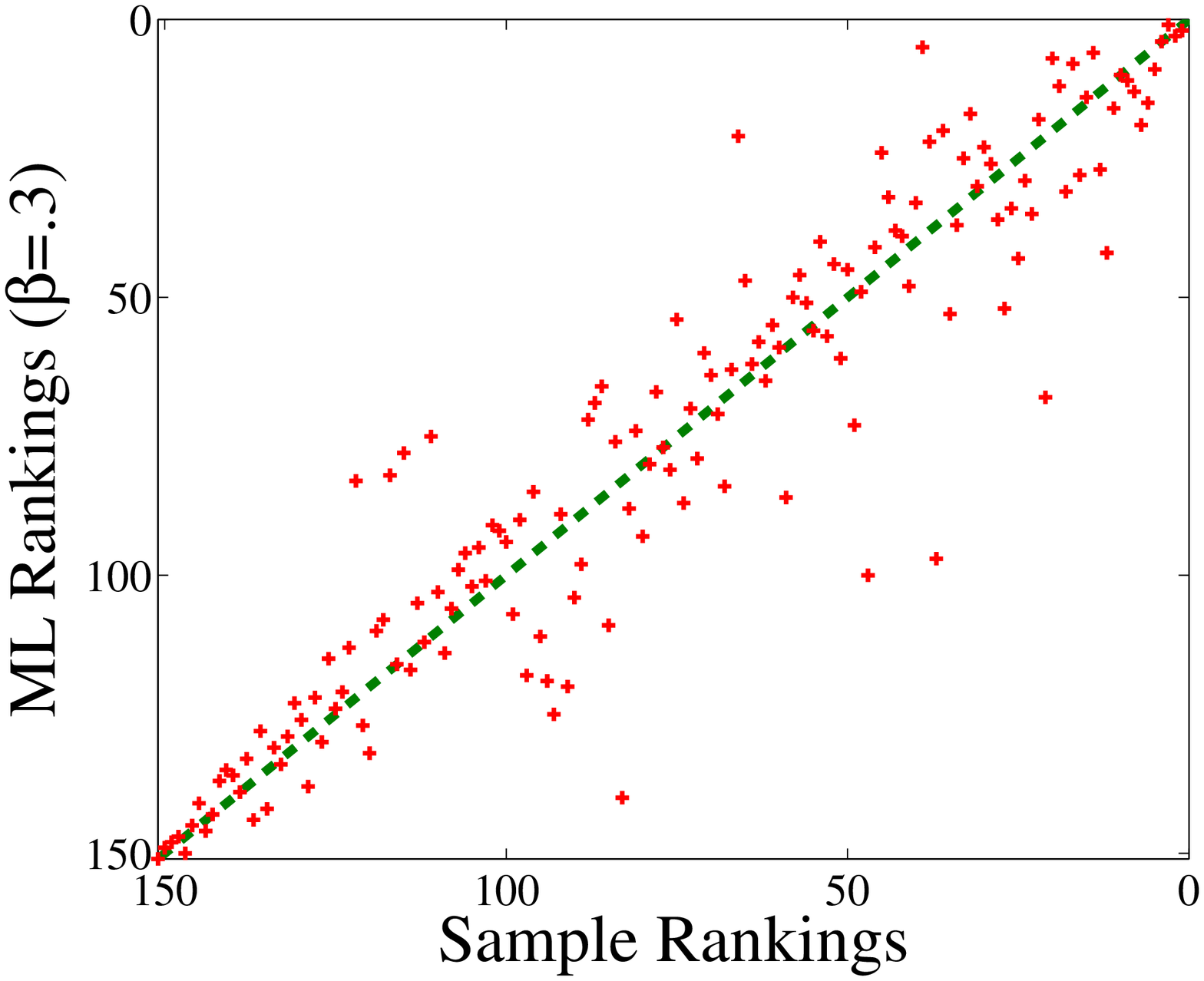}
\par\end{centering}}
\subfloat[]{\begin{centering}
\includegraphics[width=0.18\textwidth,height=2.6cm]{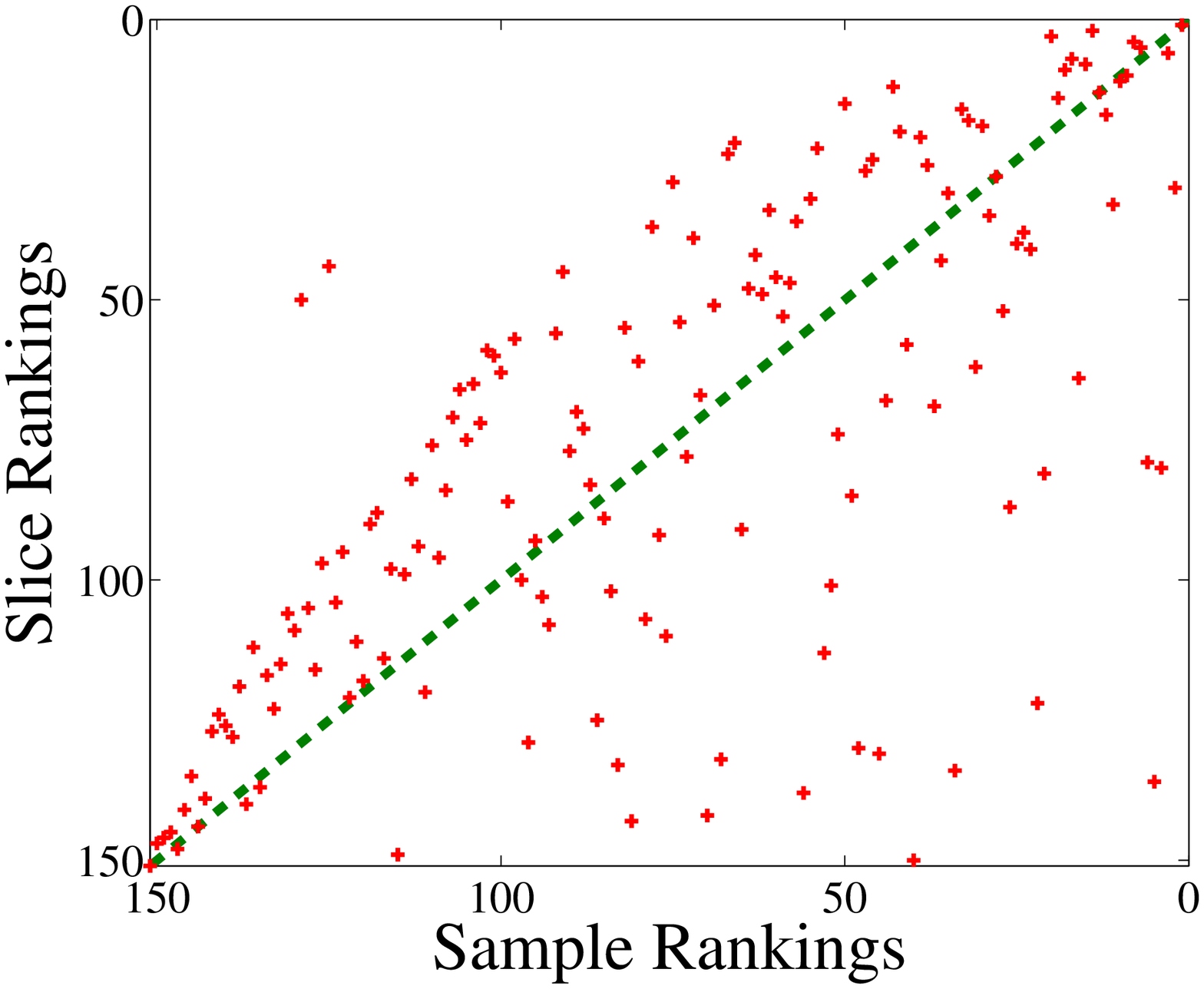}
\par\end{centering}}
\subfloat[]{\begin{centering}
\includegraphics[width=0.18\textwidth,height=2.6cm]{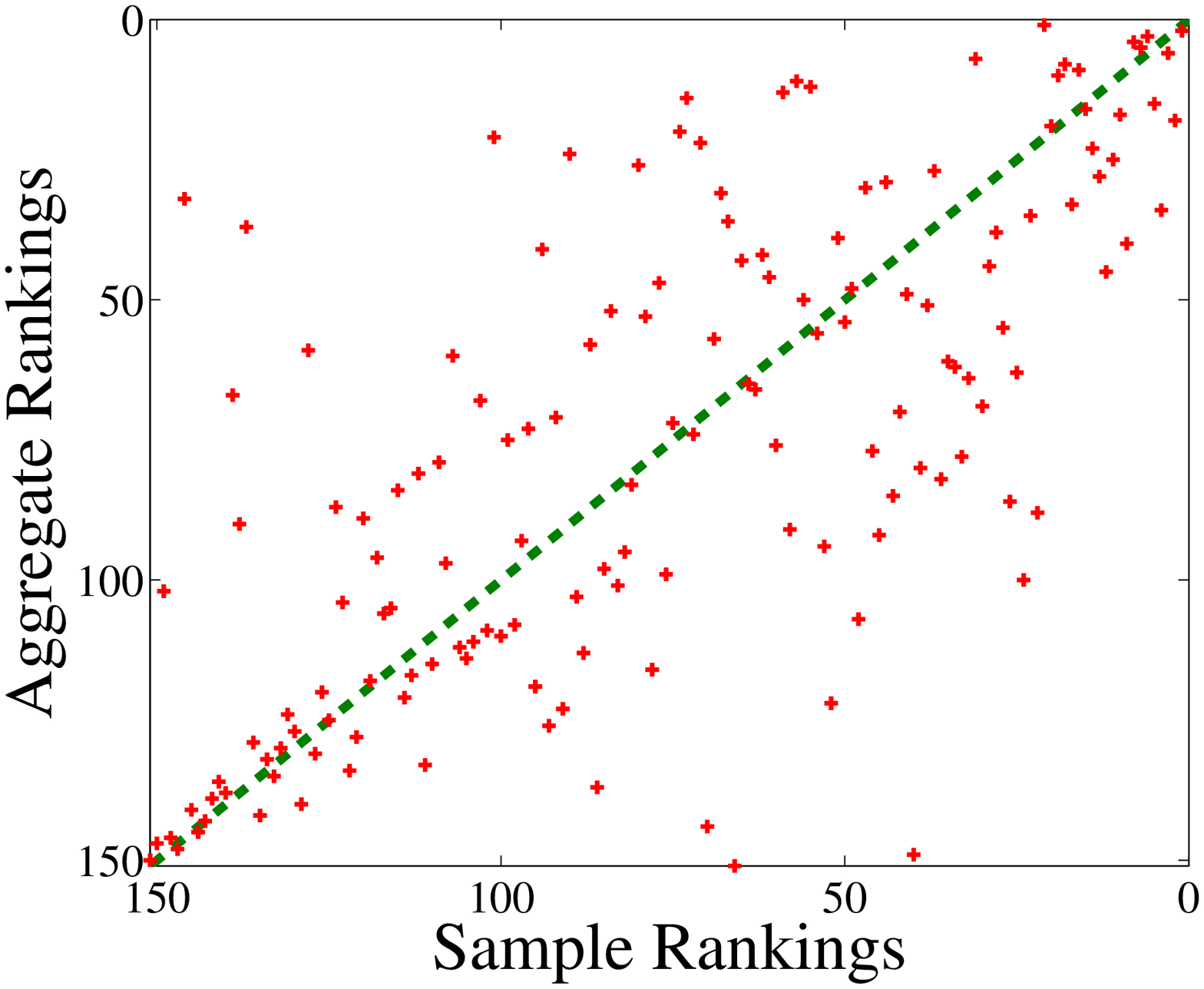}
\par\end{centering}}

\par\end{centering}
\vspace{-2.mm}

\caption{(a) Correlation between methods
for varying values of $\lambda$\label{fig:lambda-setting}. (b) Correlation of MLH with other methods as $\beta$ is varied.  (c-e) Correlations of Enron employee BCRs across methods, for the time segment ending August $24^{th}$, $2001$ \label{fig:aug-24}}
\vspace{-4.mm}

\end{figure*}

To investigate the performance of our proposed MLH and sampling methods for average path length, betweenness centrality and clustering coefficient, we compare to traditional baseline social network measures on data from Enron, DBLP, and Facebook. These datasets all consist of time-stamped {\em transactions} among people (e.g., email, joint authorship). We will use the temporal activity information to derive probabilities for use in our methods, and evaluate our measures at multiple time steps to show the evolution of measures in the three datasets.

\subsection{Datasets}

For our analysis we first use the Enron dataset \cite{enroncleaned}. The advantage to this dataset is that it allows us to understand the effects of our probabilistic
measures because key events and central
people have been well documented \cite{EnronTimeline}. 
We consider the subset of the data comprised of the emails
sent between employees,
resulting in a dataset with 50,572 emails among 151 employees. 

Our second dataset is a sample from the DBLP computer science citation database.  We considered the set of authors who had published more than 75 papers in the timeframe 1967-2006, and the coauthor relationships between them.  The resulting subset of data consisted of 1,384 nodes, with 23,748 co-authors relationships.  

Our third dataset is from the Purdue University Facebook network. 
Specifically we consider one year's worth of wall-to-wall postings between users in the class of 2011 subnetwork.  The sample has 2,648 nodes with 59,565 messages.

\subsection{Methodology}

We compare four network measures for each timestep $t$ in each dataset. When evaluating at time $t$, each method is able to utilize the graph edges that have occurred up to and including $t$. 
As baselines, we compare to (1) an \emph{aggregate} method, which
at a particular time $t$ computes standard measures for discrete graphs (e.g., BCR) on the union of edges that have occurred up to and including $t$, and (2) a time \emph{slice} method, which again computes the standard measures, but
only considers the set of edges that occur within the time window $[t-\delta, t]$. 
For the Enron and Facebook, we used $\delta=14$ days and for DBLP, we considered $\delta=1$  year. 

We then compare to the sampling and MLH measures. For both the probabilistic methods, we need a measure of relationship strength to use as probabilities in our model.
Although any notion of relationship strength can be substituted at this
step, in this work we utilize a measure of relationship strength based on
decayed message counts. 
More specifically, we define two separate
and distinct notions of connection between nodes: \emph{edges} and \emph{messages}. We define
an edge $e_{ij}$ to be the unobservable probabilistic connection between two nodes,
indicating whether the nodes have an active relationship. This is in contrast to messages: a message $m_{ij}$ is a
concrete and directly measurable communication between two nodes $v_{i}$
and $v_{j}$, such as a wall posting or email, occurring at a specific time, which we denote $t( m_{ij} )$.
We define the probability of of nodes $v_{i}$ and $v_{j}$ having
an \emph{active} relationship at the current timestep $t_{now}$, based on observing a
message at time $t(m_{ij})$, to be the exponential decay of a
particular message:
\vspace{-2.mm}
{\small
\[P\left(e_{ij}^{t}|m_{ij}\right)
=
\mbox{Exp}\left(m_{ij}|t_{now},\lambda\right)=\exp\left\{- \frac{1}{\lambda}\left(t_{now}-t\left(m_{ij}\right)\right)\right\} \] }
Note that the \emph{scaling} parameter $\lambda$ refers to the
adjustment of the basic time unit (e.g. 7 days to 1 week), not the
\emph{rate} parameter which defines the exponential probability density function, which in
this case is $1$.  This allows for assigning a
probability of $1$ to the case when $t\left(m_{ij}\right) = t_{now}$, but it also
assigns reasonable probabilities (i.e., slows the decay) for messages that
happened in the recent past, which could still indicate active relationships.

Now, we assume we have $k$ messages between $v_i$ and $v_j$, and any
of the messages $m_{ij}^1, \dots, m_{ij}^k$ can contribute to the
relationship strength, which is defined to be 1 minus the probability
that none of them contribute:
\begin{eqnarray*}
P\left(e_{ij}^{t}|m_{ij}^{1},\dots,m_{ij}^{k}\right)\!&\!=\!&\!1 -
\prod_{k}\left(1-\mbox{Exp}\left(m_{ij}^{k}|t_{now}\right)\right)\\
\end{eqnarray*}
\vspace{-8mm}

In order to choose a scaling parameter $\lambda$ for the exponential decay,
we measured the average correlation from the sampling method BCR against the time slice ranking and aggregate method for each
Enron employee, for different values of $\lambda$ (see Figure
\ref{fig:lambda-setting}.a).  Note that a $\lambda$ close to 0 corresponds to
`forgetting' a transaction quickly and is highly correlated with the slice
method, while a large $\lambda$ corresponds to 
`remembering' a transaction for a long time, giving it high correlation with the
aggregate method.
In order to
balance between short term change and long term trends we set $\lambda$ to a
`middle ground' with $\lambda=28$ days. This applies to both the Enron and
Facebook datasets. For DBLP, where we evaluate yearly, $\lambda$ is set to $2$
years to keep the ratio between time slice and $\lambda$ consistent
between Facebook, Enron, and DBLP.

In order to choose a value for the $\beta$ parameter in the MLH method,
we measured the average correlation of the BCR from
the MLH method and compared them to the sampling, aggregate, and slice
rankings for different values of $\beta$. We can see in Figure
\ref{fig:aug-24}.b that as long as $\beta$ is non-zero, it has minimal effect
on the correlations. For the experiments reported in this paper, we set $\beta=.3$.  Note that omission of the prior (i.e., $\beta=1$) in
will make the MLH paths similar to the slice paths, with
added paths between vertices which are disjoint in a particular time slice.

The final parameter setting is the number of samples to consider in each of sampling-based measures. Earlier we discussed how we are computing the
joint instances of shortest paths, and that the bound by \cite{sampleprobpaths} does not hold.  
Due of this, we exploit the small size of the Enron dataset and take 10,000 samples; however, 
with the two larger graphs we use a smaller sample size of 200 in order to make the experiments tractable.

\subsection{Method Correlations on Enron Data}
In order to illustrate the differences between the four methods, we
analyze their respective BCR on the Enron data for the time window ending August
14$^{th}$, 2001. 
Figure \ref{fig:aug-24}.c-e shows the correlations of employee
BCR across a pair of methods:
points on the diagonal green
line indicate `perfect' correlation between the rankings of two methods.

Figure \ref{fig:aug-24}.c shows that the MLH method closely matches the sampling
method, with only a few nodes varying from the diagonal.
However, a large number of nodes that the sampling method determines to have high centrality 
are missed by the slice method, due to the slice's inability to see transactions that occurred prior to the evaluation time window.
Additionally, we note that August $14^{th}$, 2001 is relatively late in the Enron timeline, which results in the aggregate method
having little correlation with the sampling method, since the more recent changes are washed out by past transactions in the aggregate approach.

\subsection{Local Trend Analysis}
\subsubsection{Lay and Skilling}

\begin{figure}[t]
\vspace{-4.mm}
\begin{centering}
\subfloat[]{\begin{centering}
\includegraphics[width=0.49\columnwidth,height=2.6cm]{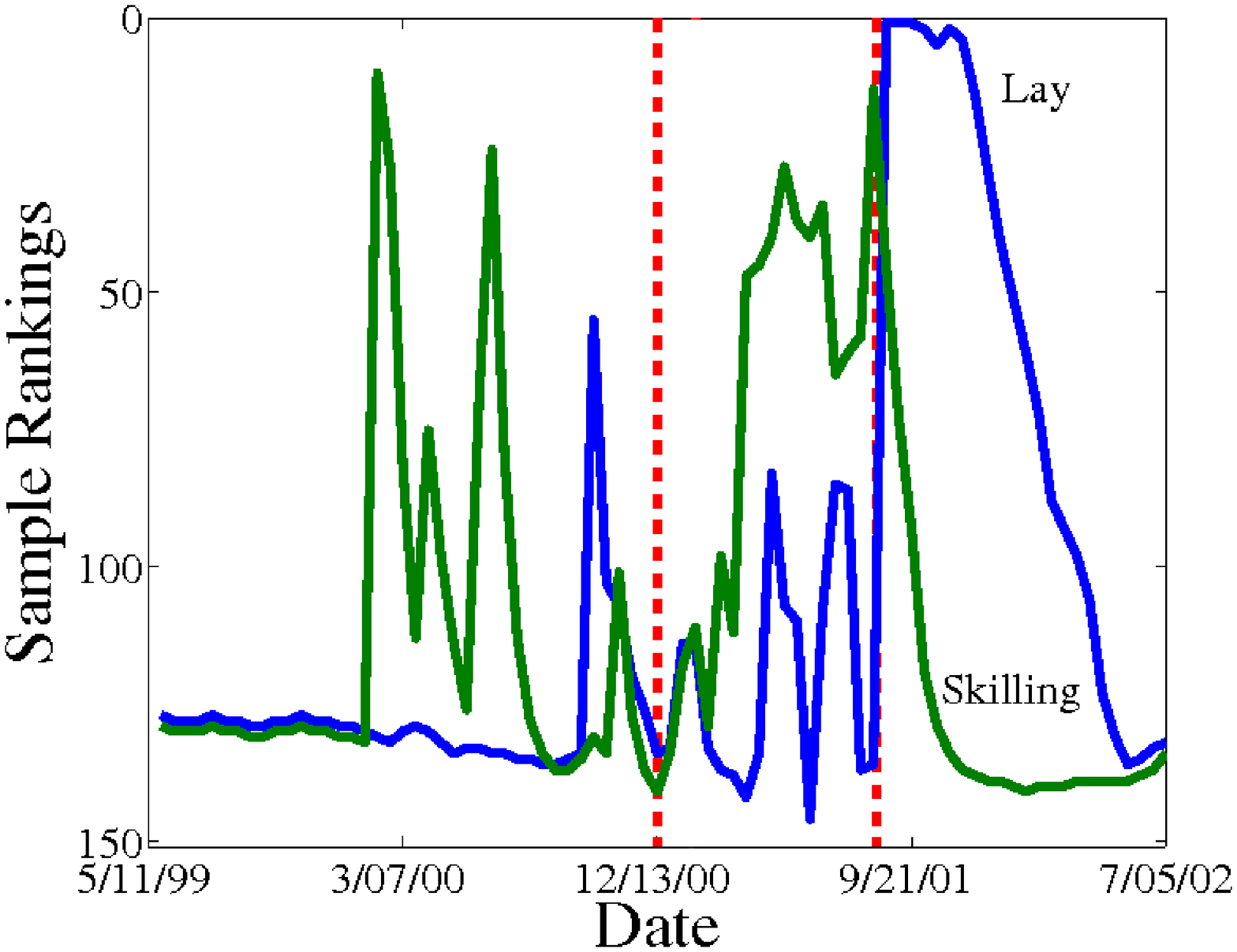}
\par\end{centering}

}\subfloat[]{\begin{centering}
\hspace{-4.mm}
\includegraphics[width=0.49\columnwidth,height=2.6cm]{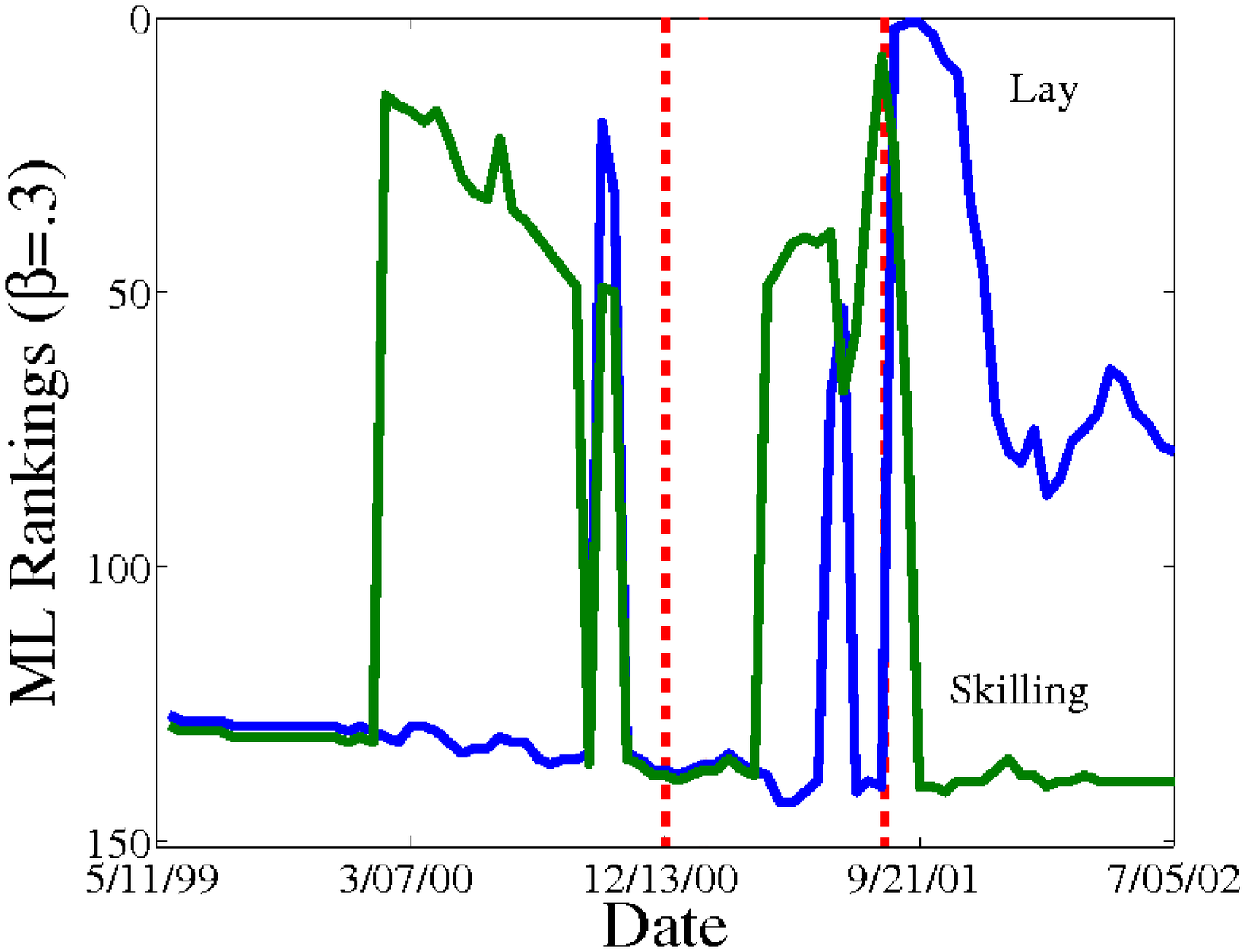}
\par\end{centering}

}
\par\end{centering}

\vspace{-4.mm}
\begin{centering}
\subfloat[]{\begin{centering}
\includegraphics[width=0.49\columnwidth,height=2.6cm]{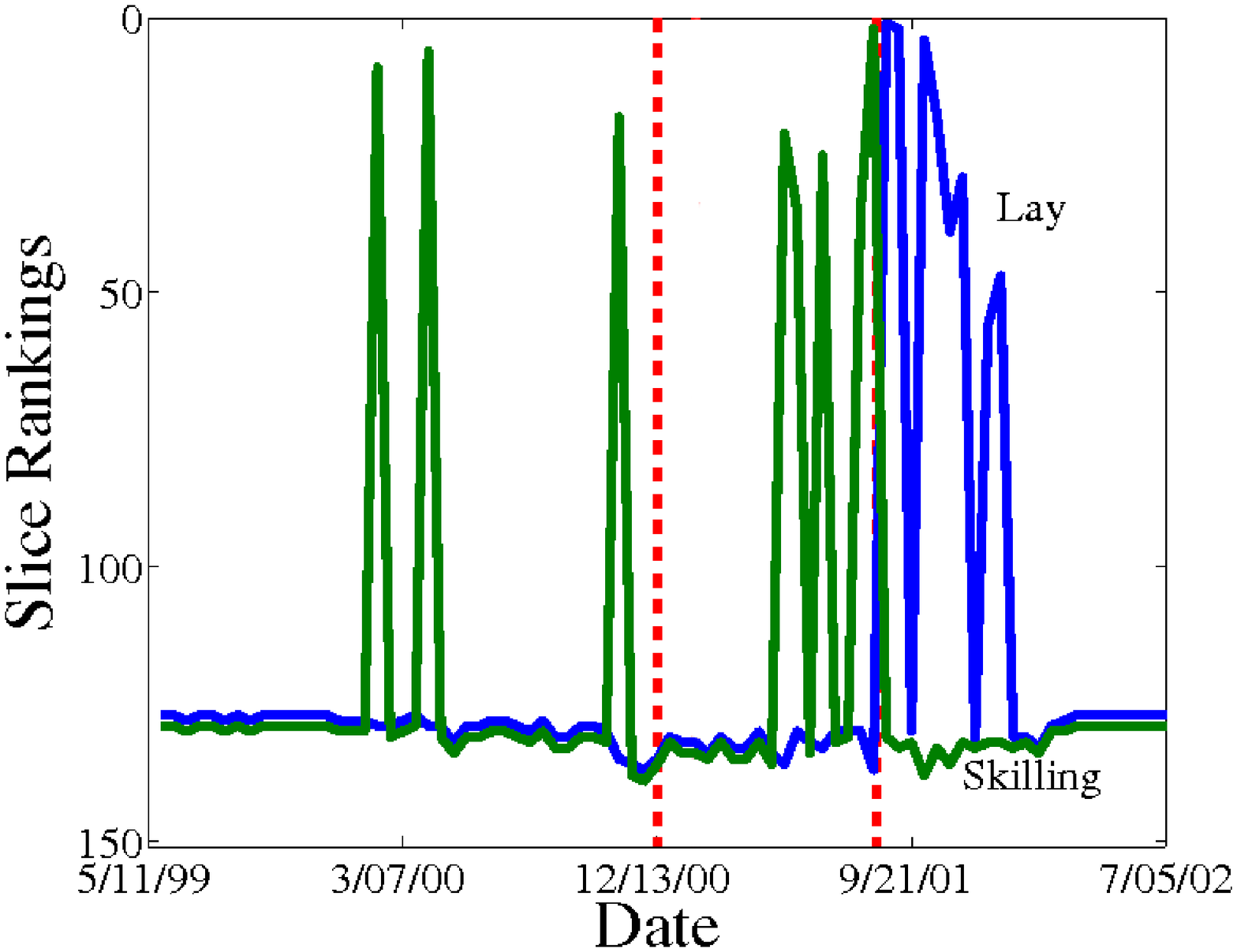}
\par\end{centering}

}\subfloat[]{\begin{centering}
\hspace{-4.mm}
\includegraphics[width=0.49\columnwidth,height=2.6cm]{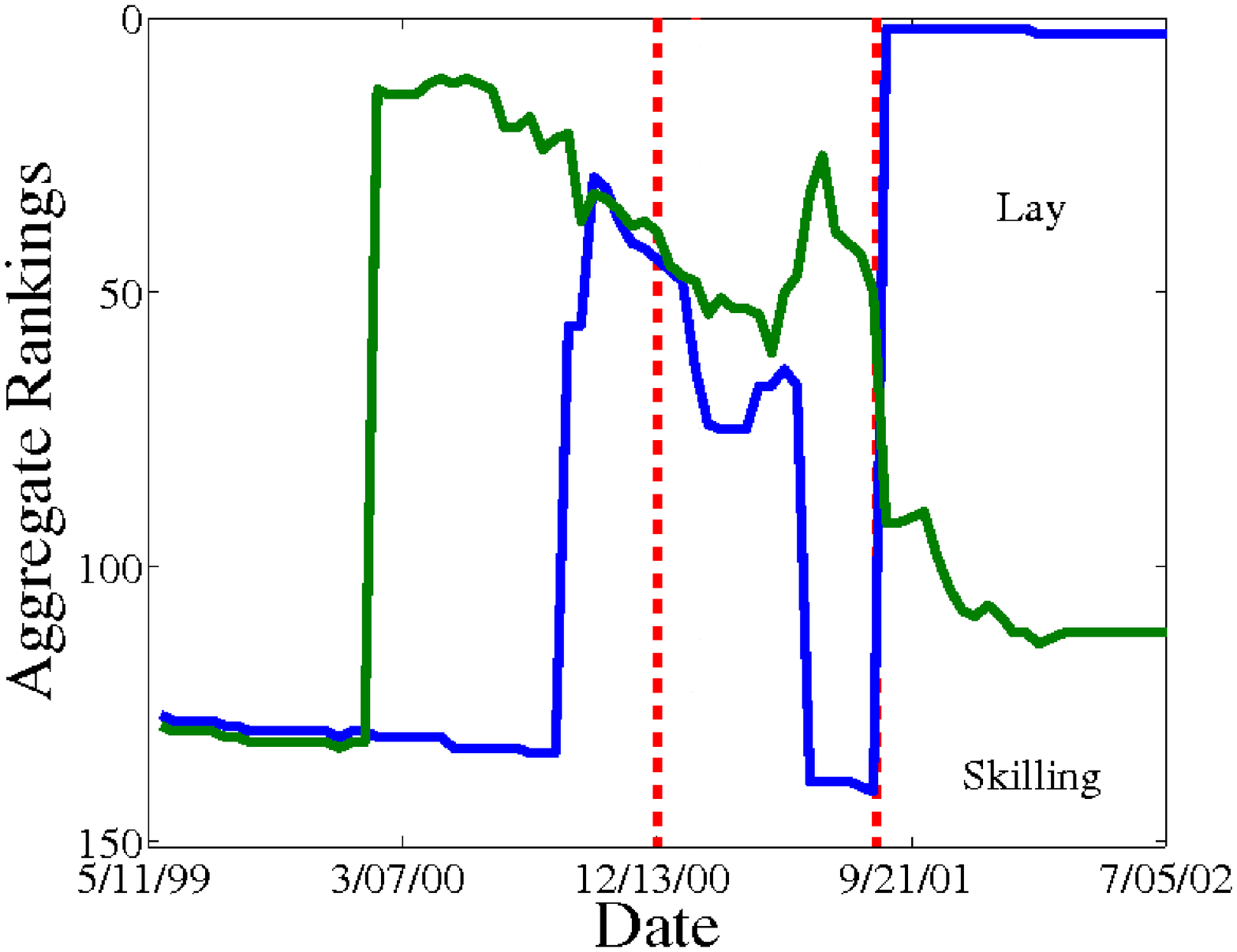}
\par\end{centering}

}
\par\end{centering}

\vspace{-2mm}
\caption{BCR of Lay and Skilling over time. Red lines indicate Skilling's CEO
announcement and resignation.\label{fig:layskilling}}
\vspace{-4.mm}
\end{figure}
Here, we analyze 
two key figures at Enron: Kenneth Lay and
Jeffery Skilling.  These two were central to the Enron scandal---as first Lay, then Skilling, and then Lay again, assumed the position of CEO.
We can analyze the BCR for Lay and Skilling
during these transition periods, as we expect large changes to affect both of them. 

The first event we consider (marked by a vertical red line in Figure~\ref{fig:layskilling}) is {\em December 13$^{th}$ 2000}, when it was announced that Skilling would assume
the CEO position at Enron, with Lay retiring but remaining as a chairman
\cite{EnronTimeline}. In Figure \ref{fig:layskilling}.a,
both the sampling method and the MLH method identify a spike
in BCR for both Lay and Skilling directly before the announcement.
This is not surprising, as presumably Skilling and Lay were informing
the other executives about the transition that was about to be announced.

The time slice method (\ref{fig:layskilling}.c) produces no change in Lay's BCR,
despite his central role in the transition.
Skilling shows a few random spikes of BCR, which illustrates the variance associated with using the time slices. The aggregate model (\ref{fig:layskilling}.d) fails to
reduce Skilling's BCR to the expected levels following the announcement---this is fairly early in time and
we are already seeing the aggregate method's inability to track current events based on its union of all past transactions.
Both the sampling method and the MLH
methods capture this; MLH has him return to an extremely low
centrality, while sampling has fairly low with some variance.

\begin{figure}[t]
\vspace{-4.mm}
\begin{centering}

\subfloat[]{\begin{centering}
\hspace{-4.mm}
\includegraphics[width=0.49\columnwidth,height=2.6cm]{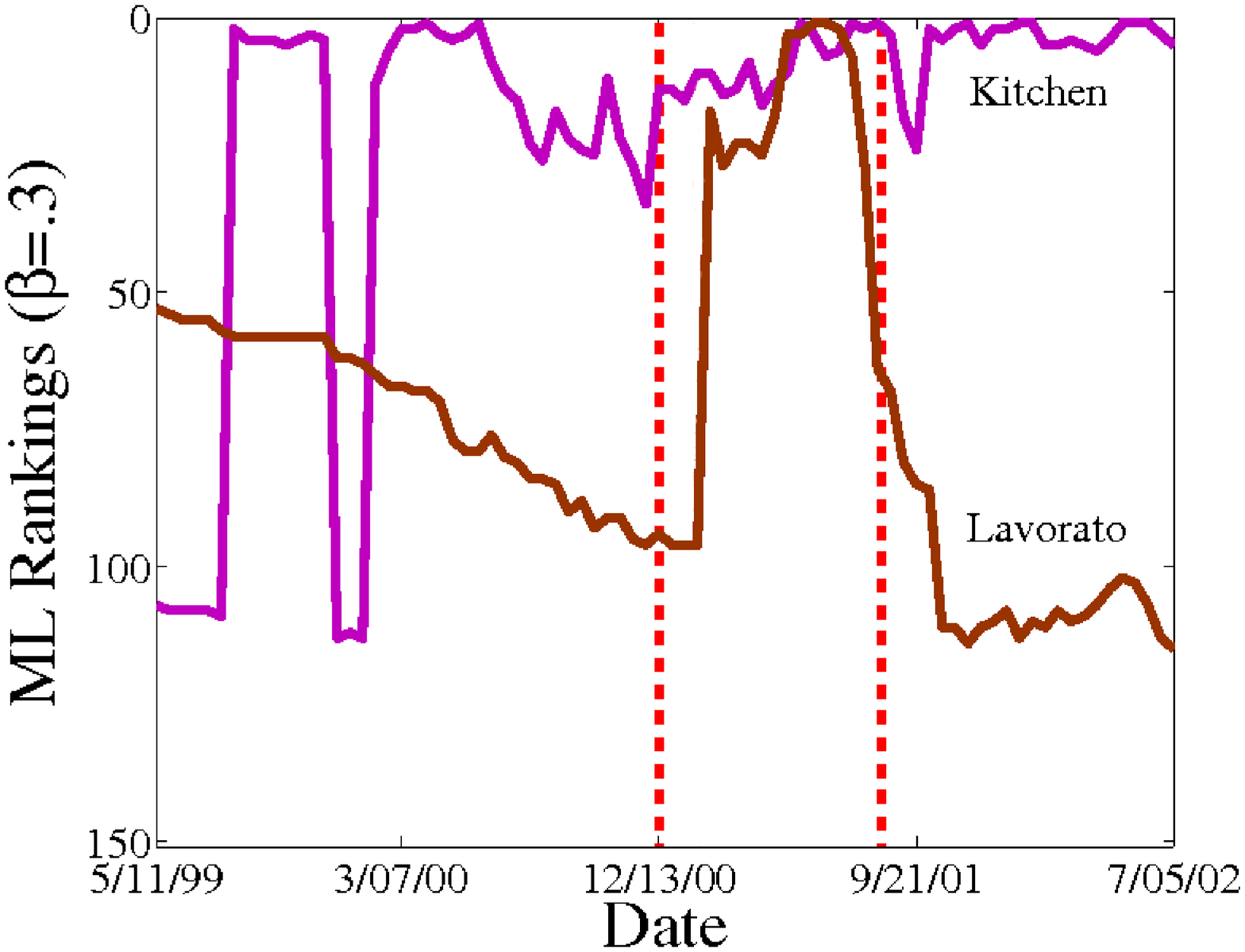}
\par\end{centering}
}\subfloat[]{\begin{centering}
\hspace{-4.mm}
\includegraphics[width=0.49\columnwidth,height=2.6cm]{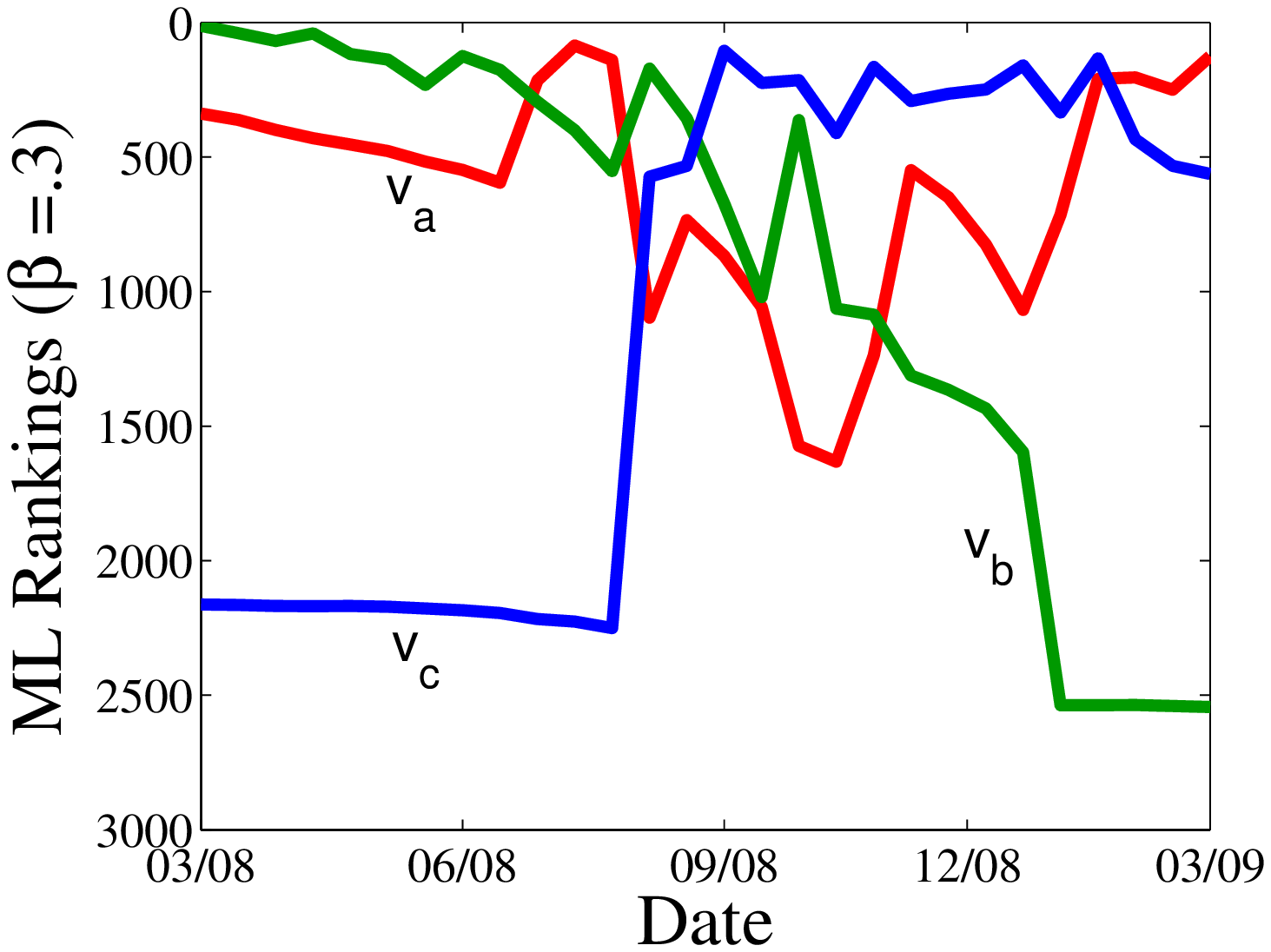}
\par\end{centering}
}
\par\end{centering}

\vspace{-4.mm}

\begin{centering}

\subfloat[]{\begin{centering}
\hspace{-4.mm}
\includegraphics[width=0.49\columnwidth,height=2.6cm]{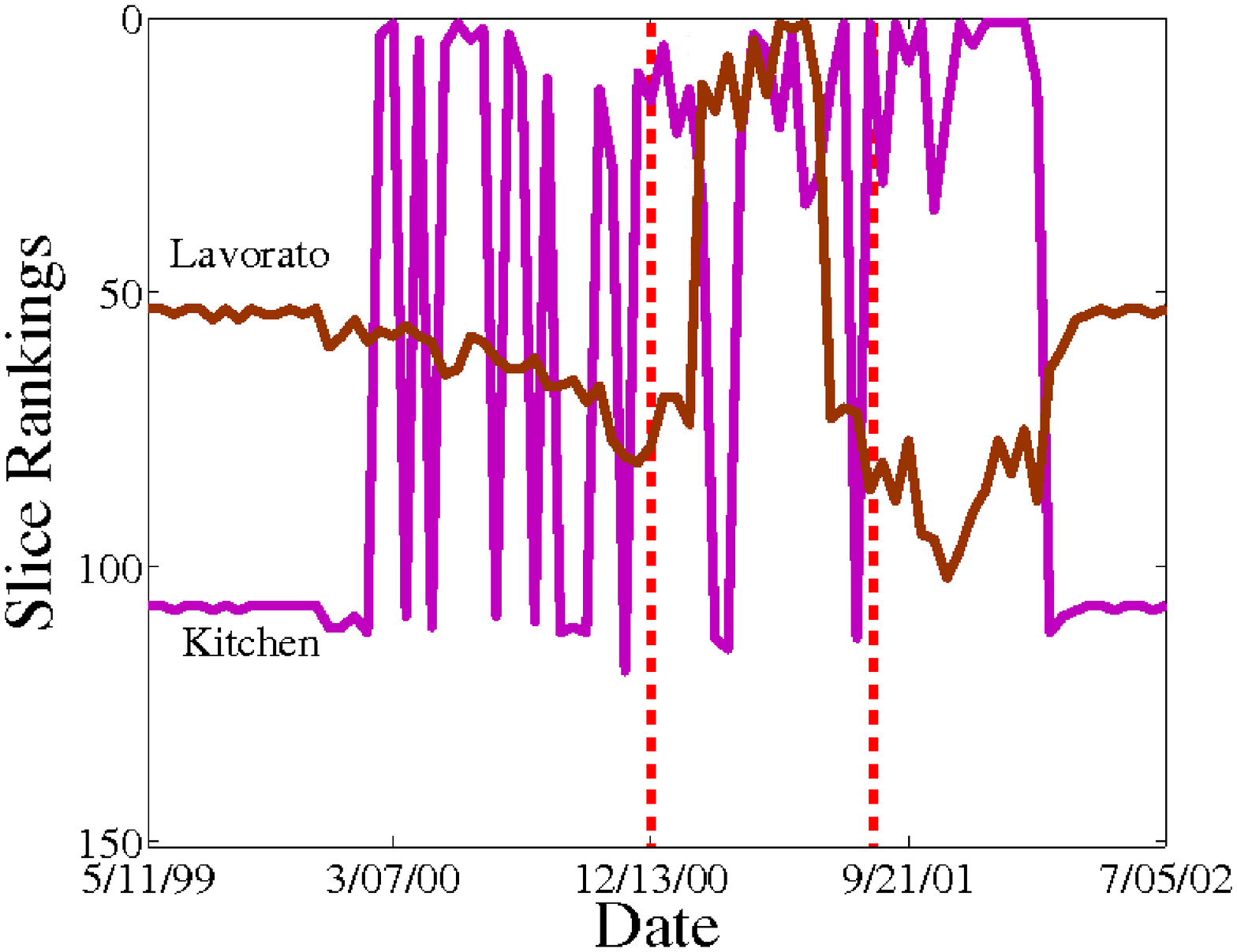}
\par\end{centering}
}\subfloat[]{\begin{centering}
\hspace{-4.mm}
\includegraphics[width=0.49\columnwidth,height=2.6cm]{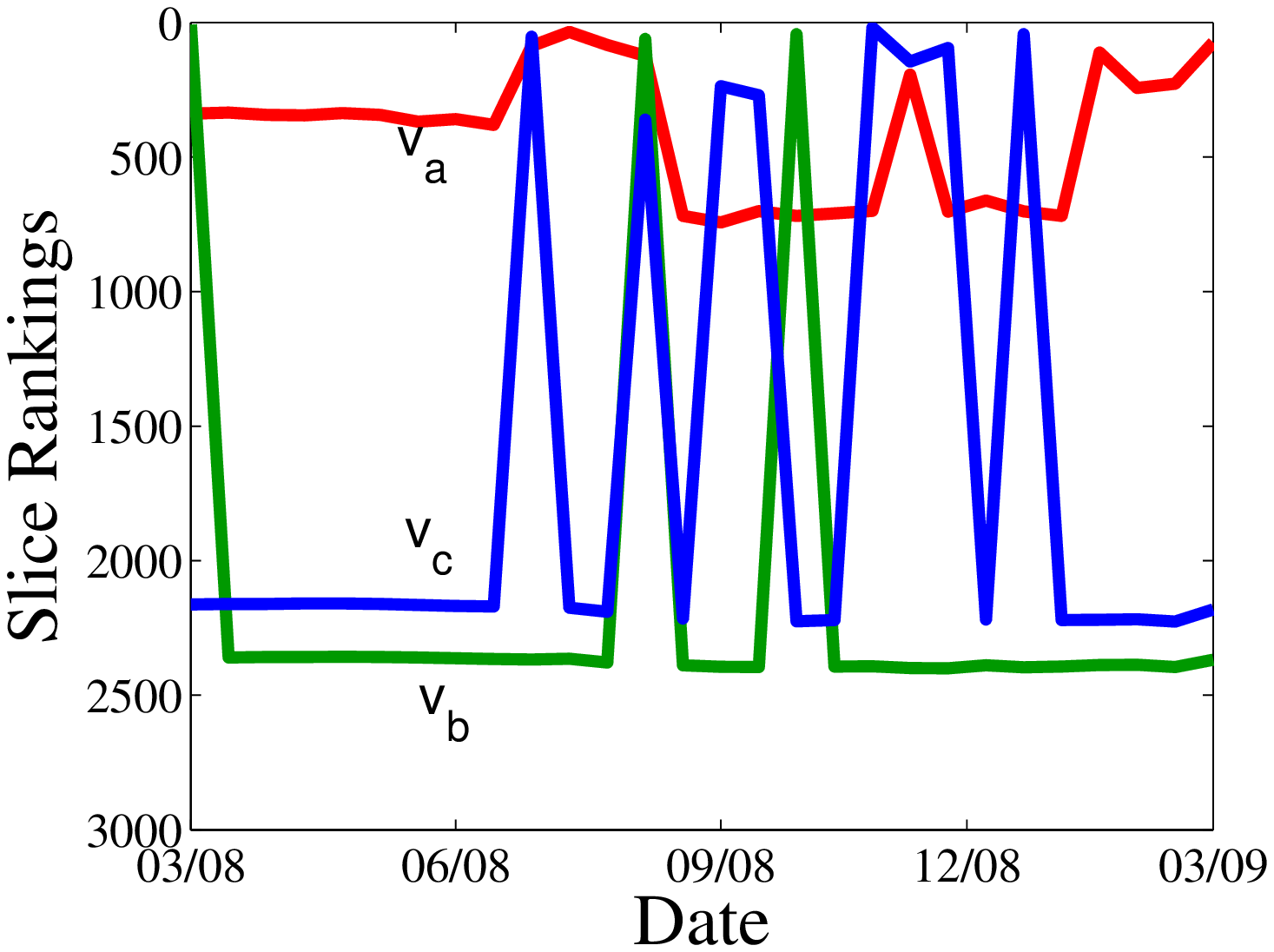}
\par\end{centering}
}

\par\end{centering}

\vspace{-4.mm}

\begin{centering}
\subfloat[]{\begin{centering}
\hspace{-4.mm}
\includegraphics[width=0.49\columnwidth,height=2.6cm]{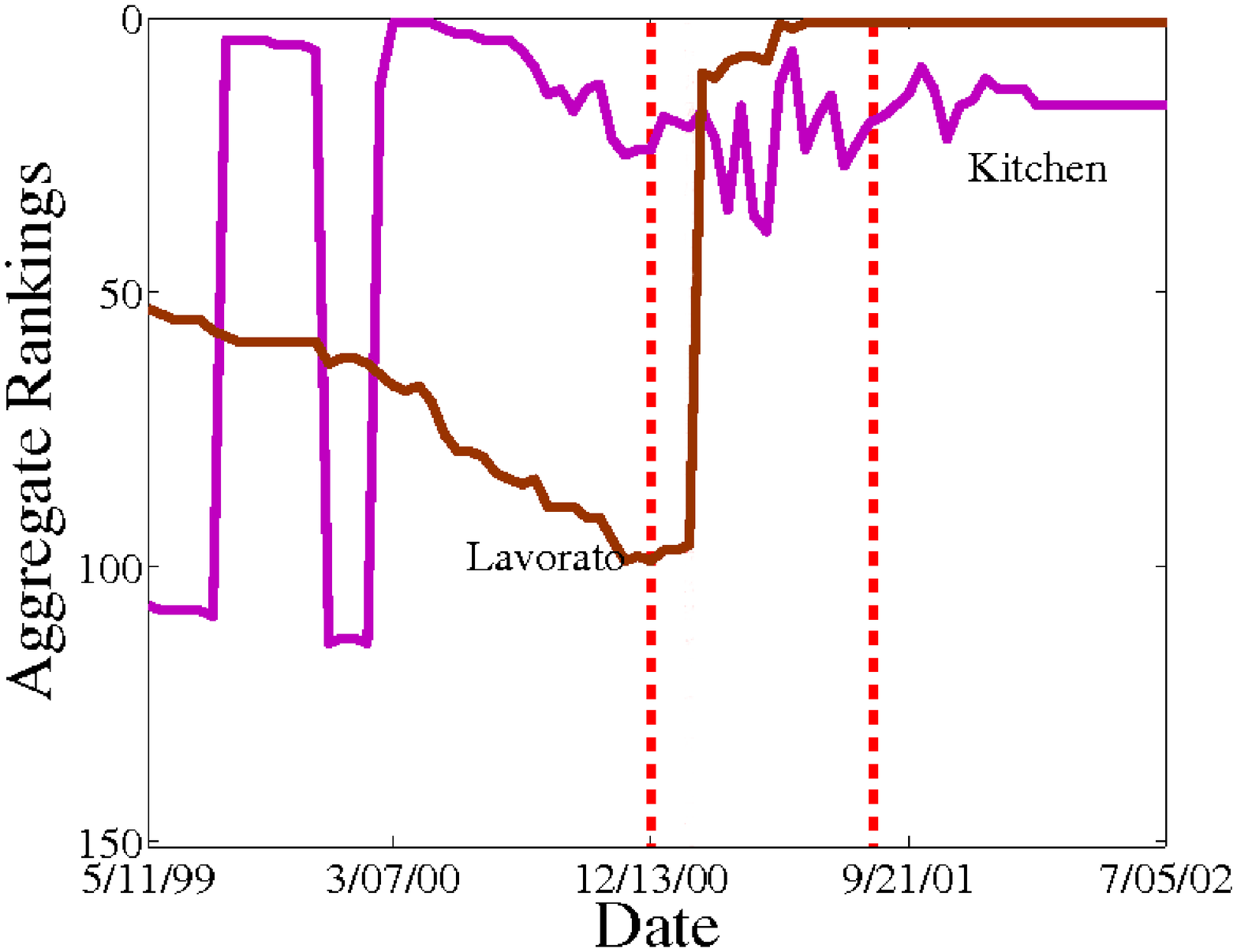}
\par\end{centering}
}\subfloat[]{\begin{centering}
\hspace{-4.mm}
\includegraphics[width=0.49\columnwidth,height=2.6cm]{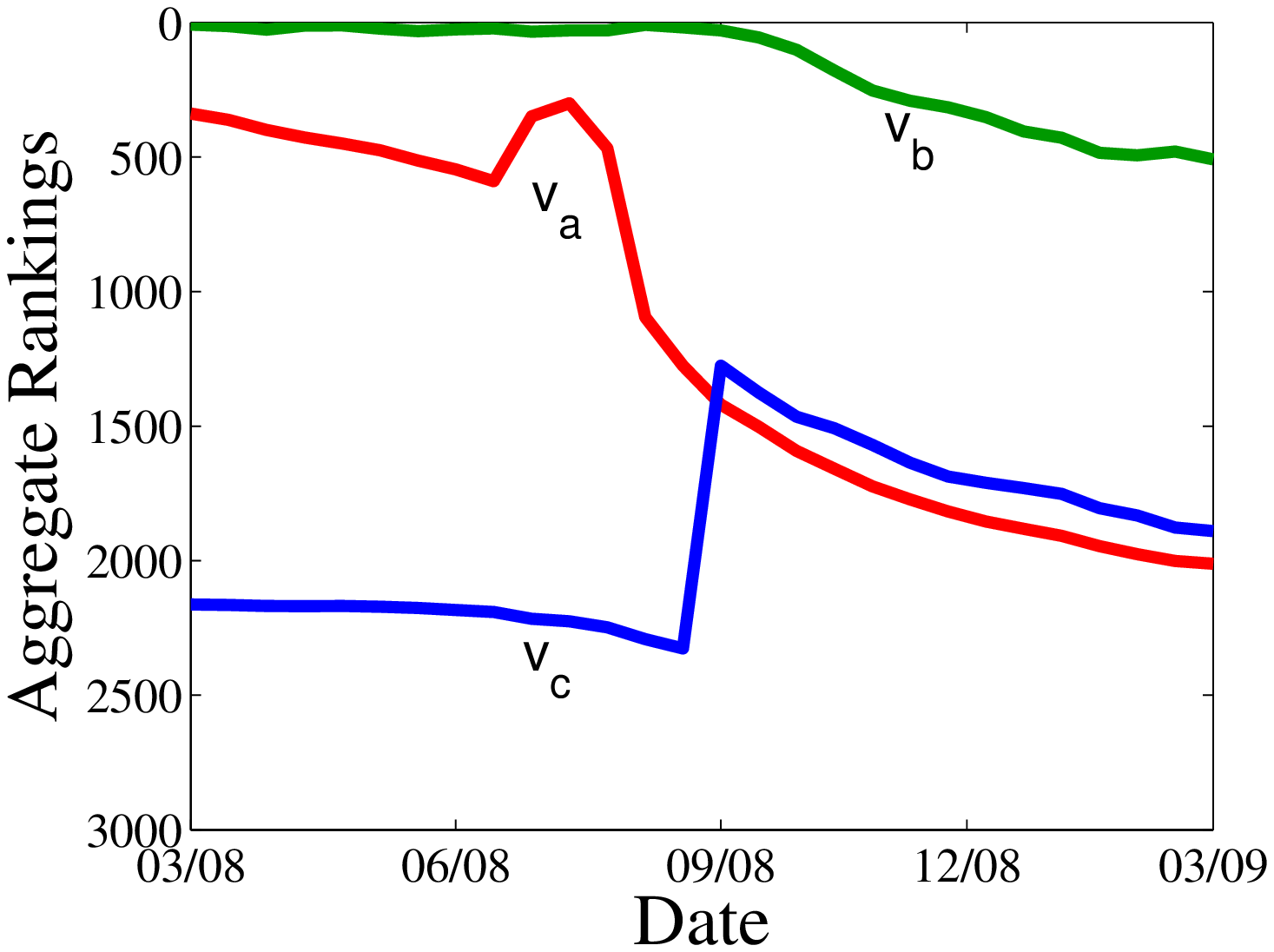}
\par\end{centering}

}

\par\end{centering}

\vspace{-2mm}
\caption{(a,c,e) BCR of Kitchen and Lavorato.\label{fig:kitchenlavorato}  (b,d,e) BCR for 3
nodes in the Purdue Facebook network.}

\vspace{-4.mm}
\end{figure}

The second event we consider (marked by the 2nd vertical red line in Figure~\ref{fig:layskilling}) is {\em August 14$^{th}$ 2001}, when, seven months after initially taking the CEO position, Skilling approached
Lay about resigning~\cite{EnronTimeline}. During the entirety of
Skilling's tenure, we see that Lay has a slight effect on the sample
rankings but is not what would be considered a `central' node. Not
surprisingly, Skilling has a fairly high centrality during his time
as CEO; both the sampling method and MLH method capture this.

Prior to the announcement of Lay's takeover as CEO, the slice method
still had no weight on him, despite his previous involvement with
the first transition. Also, we note that the sampling, MLH,
and slice methods all agree that after Lay's initial spike from the
Skilling resignation, he resumes having a lower centrality, which the aggregate
method misses.    In
general, the sampling method seems to mirror the slice method, albeit with less variance,
but it not as smooth as the MLH method, indicating the utility of considering
most probable paths.

\subsubsection{Kitchen and Lavorato}
Next we analyze Louise Kitchen and John Lavorato, who were executives \cite{enroncleaned} for Enron Americas, which
was the wholesale trading section of Enron \cite{EnronAmericas}. They are notable because of the 
extraordinarily high bonuses they received as Enron was being investigated, and were also
found to have a high temporal betweenness centrality using the method
defined by \cite{tang-tempshortestpaths}. We can see in Figure \ref{fig:kitchenlavorato} (a,c,e)
the rankings of Kitchen and Lavorato, and can see the benefit
of using the probabilistic framework's ability to key in on centralities
at \emph{specific} times, rather than using the temporal
definition \emph{through} time proposed by \cite{tang-tempshortestpaths}. We see that
while Lavorato might have gotten a large bonus, he is \emph{only}
important during Skilling's tenure as CEO; his centrality drops noticeably
otherwise. On the other hand, Kitchen had extremely high rankings
throughout. 

Here, we see that the slice method exhibits high variability, especially with Kitchen,
while the aggregate cannot recognize Lavorato's lack of importance after Skilling's
departure.
The MLH method
is able to smoothly capture Kitchen's centrality, while keeping Lavorato important
solely during Skilling's CEO tenure.

\subsubsection{Facebook Centrality}

Unlike the Enron dataset, the Purdue Facebook dataset does not have well-established ground truths, 
where we can use the known characteristics and behaviors of particular nodes for evaluation. 
However, we can
examine aspects of a few representative nodes to 
illustrate the problems that lie with usage of the aggregate or static methods.
First, we can see from Figure \ref{fig:kitchenlavorato}.d that 
$v_a$ (red) has a consistently high ranking in the slice method, which the MLH method captures (\ref{fig:kitchenlavorato}b).  
However, this person has a declining ranking in the aggregate method, as
the aggregate is unable to capture current events---past information in the
aggregate graph results in many paths that bypass $v_a$, missing this
central node in later timesteps.

The next person we consider is denoted by $v_b$
(green).  In \ref{fig:kitchenlavorato}.d, we can see that the slice method
initially identifies this person as having high centrality, then their BCR bottoms out,
and then peaks a few times again approximately midway through the timeline. The
MLH method also initially identify $v_b$ as central, with a
degradation over time.
In contrast, the aggregate method fails to detect the inactivity
later in the timeframe and continues to give $v_b$ a high centrality ranking
throughout the entire time window.

The final person we consider is denoted by $v_c$ (blue) in Figure \ref{fig:kitchenlavorato}.  We can see in \ref{fig:kitchenlavorato}.d
that the slice method exhibits large variability for $v_c$, but that there are many slices in the middle to end of the timeframe where the node is identified 
as highly central.  The aggregate method is unaware of this activity and ranks $v_c$ at a relatively low level throughout the timeseries.  
In contrast, the MLH method is able to recognize the node's growing importance as time evolves, 
and do so much more smoothly than the slice method (\ref{fig:kitchenlavorato}.d).  In doing so, the MLH method can find instances of high centrality when both discrete methods fail.

\subsection{Global Trend Analysis}

\begin{figure}[t]
\vspace{-4.mm}

\begin{centering}
\subfloat[]{\begin{centering}
\includegraphics[width=0.49\columnwidth,height=2.6cm]{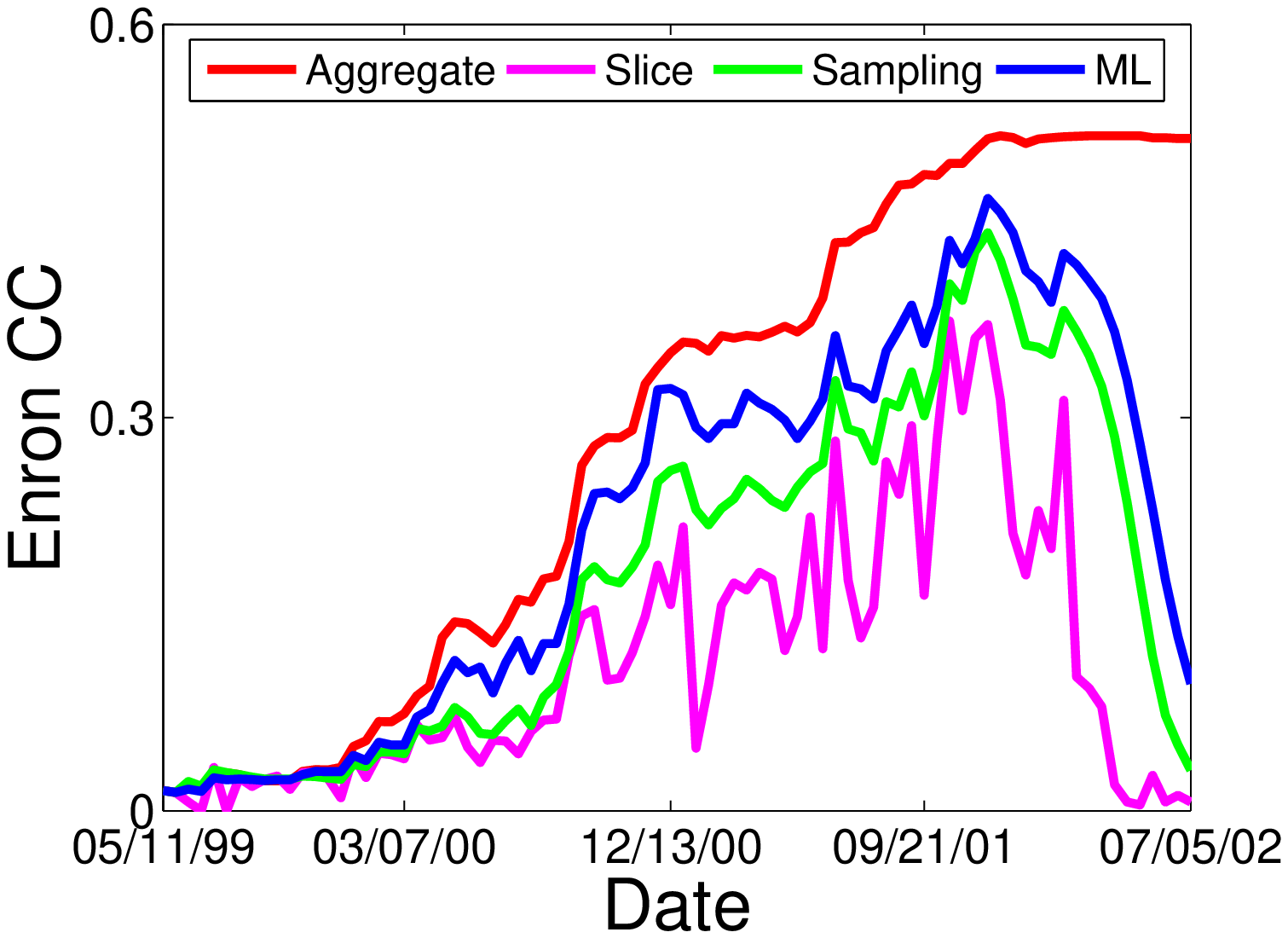}
\par\end{centering}}
\subfloat[]{\begin{centering}
\includegraphics[width=0.49\columnwidth,height=2.6cm]{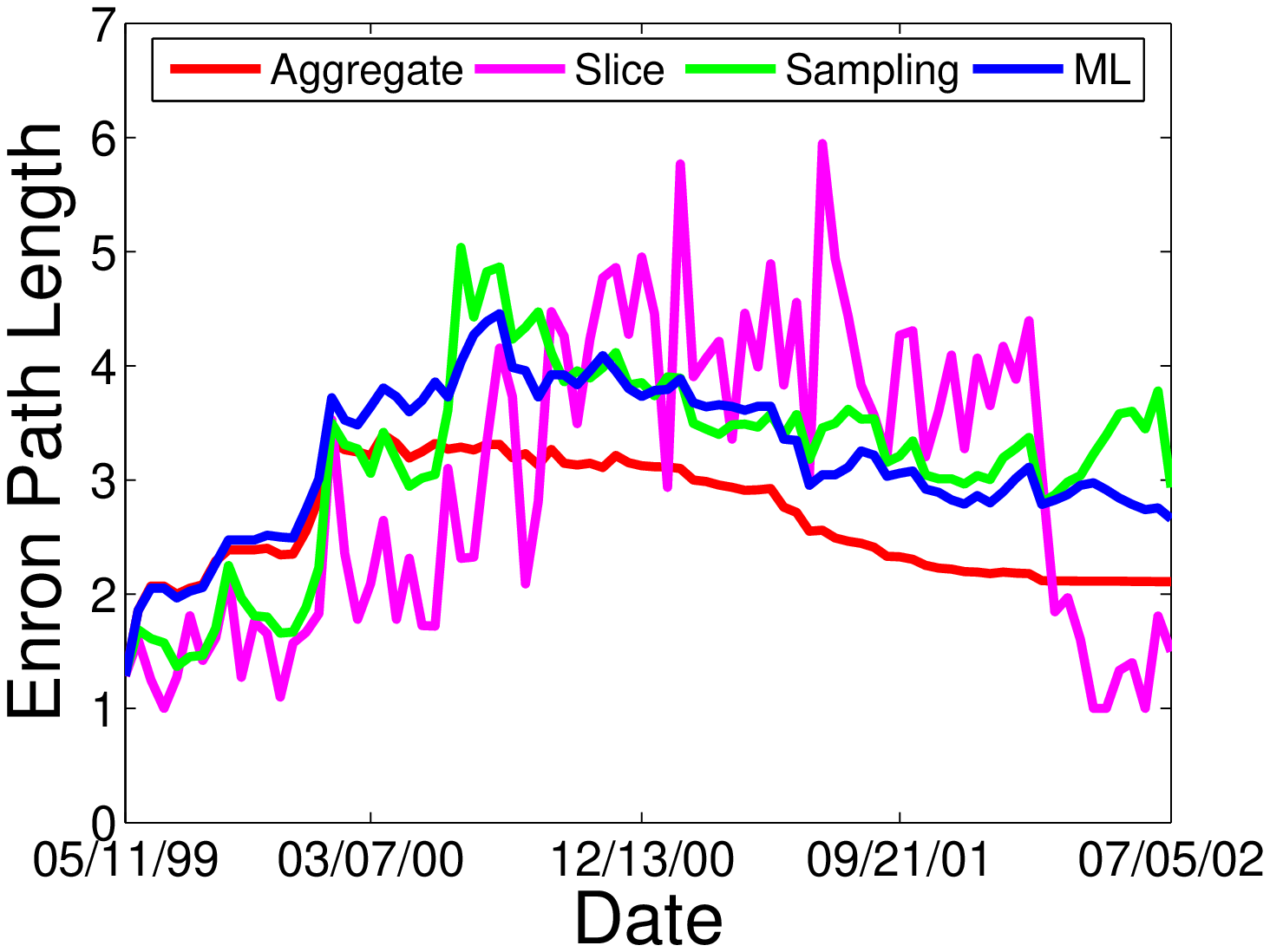}
\par\end{centering}}

\vspace{-4.mm}
\subfloat[]{\begin{centering}
\includegraphics[width=0.49\columnwidth,height=2.6cm]{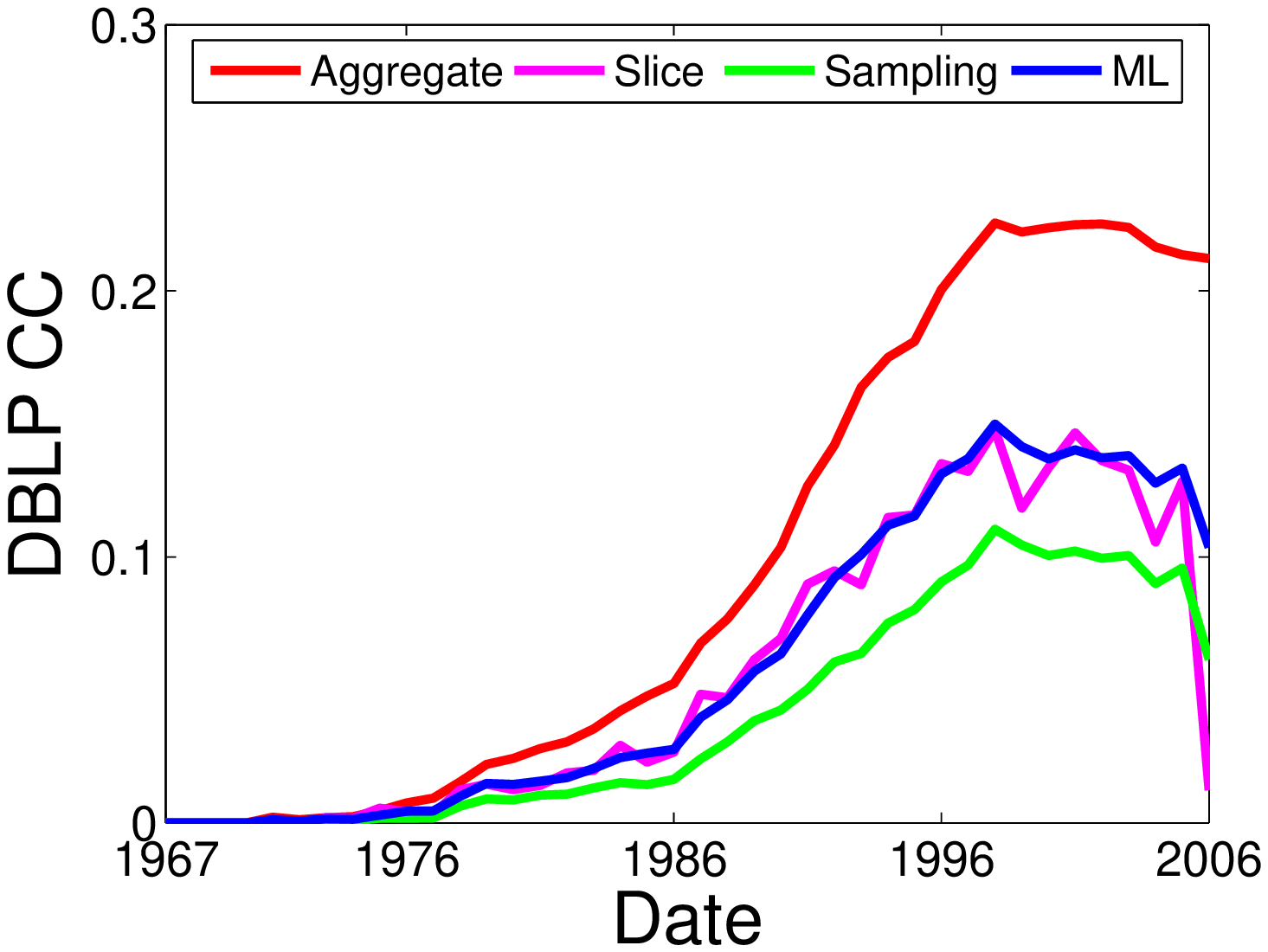}
\par\end{centering}}
\subfloat[]{\begin{centering}
\includegraphics[width=0.49\columnwidth,height=2.6cm]{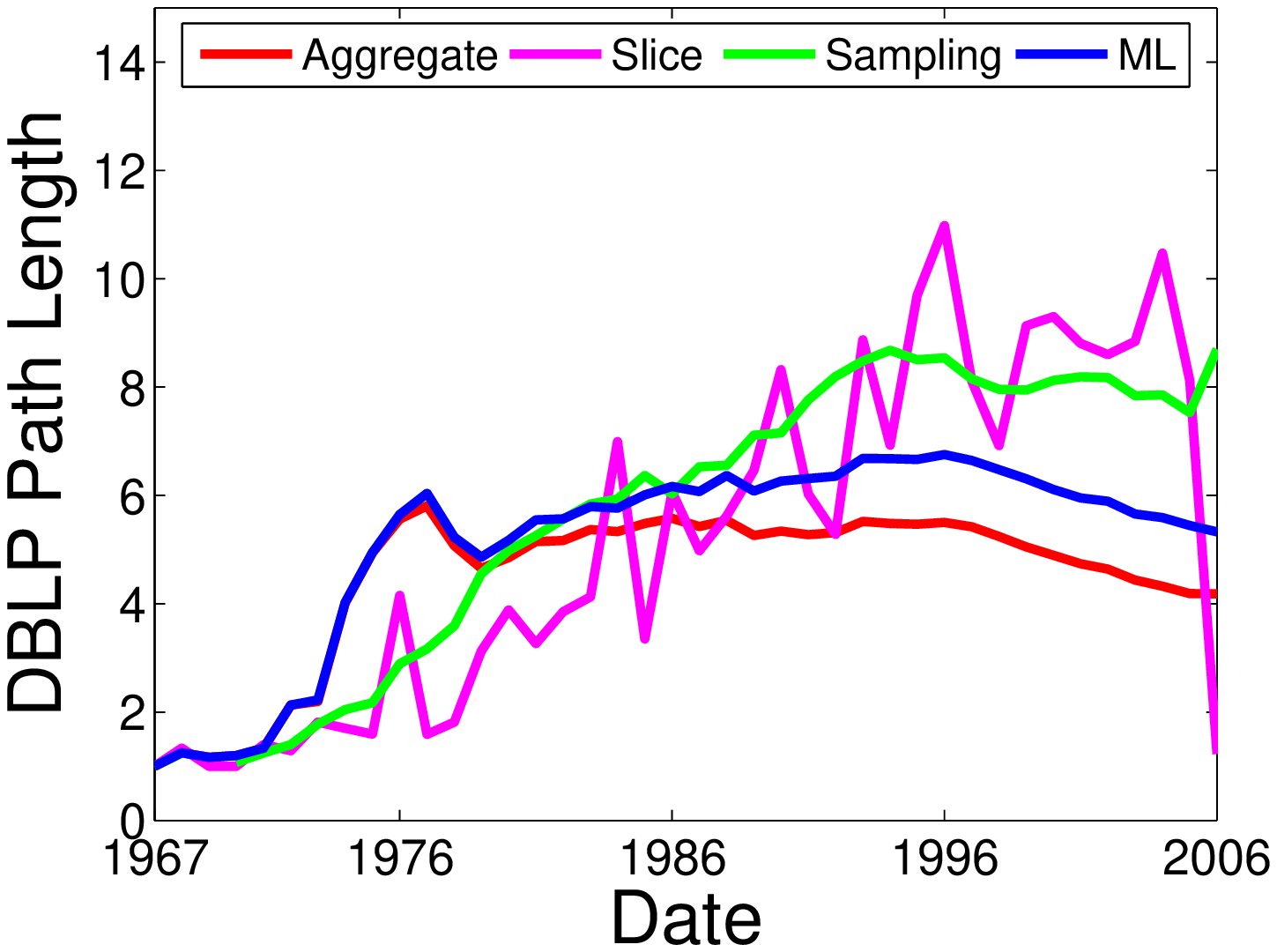}
\par\end{centering}}

\vspace{-4.mm}
\subfloat[]{\begin{centering}
\includegraphics[width=0.49\columnwidth,height=2.6cm]{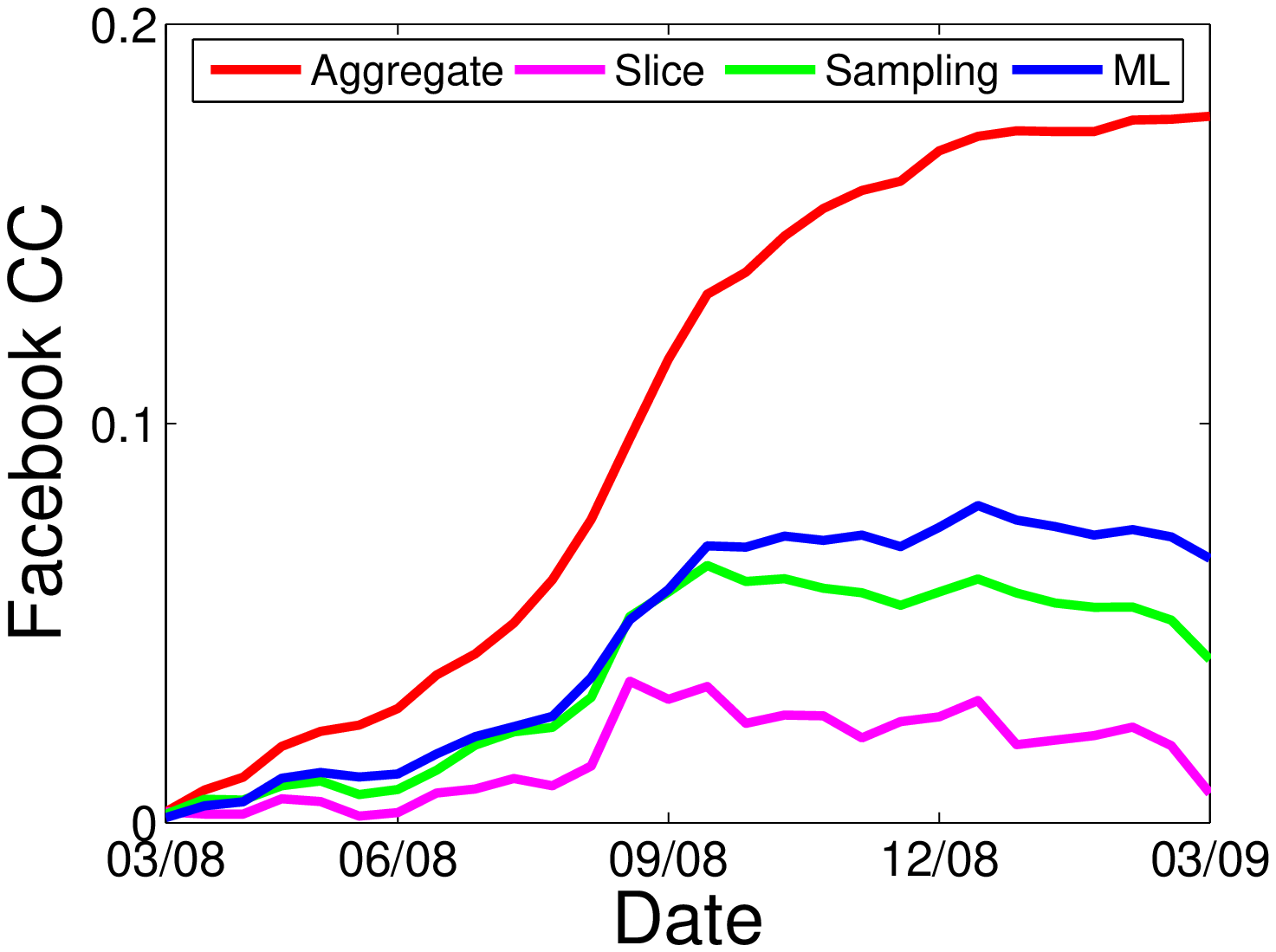}
\par\end{centering}}
\subfloat[]{\begin{centering}
\includegraphics[width=0.49\columnwidth,height=2.6cm]{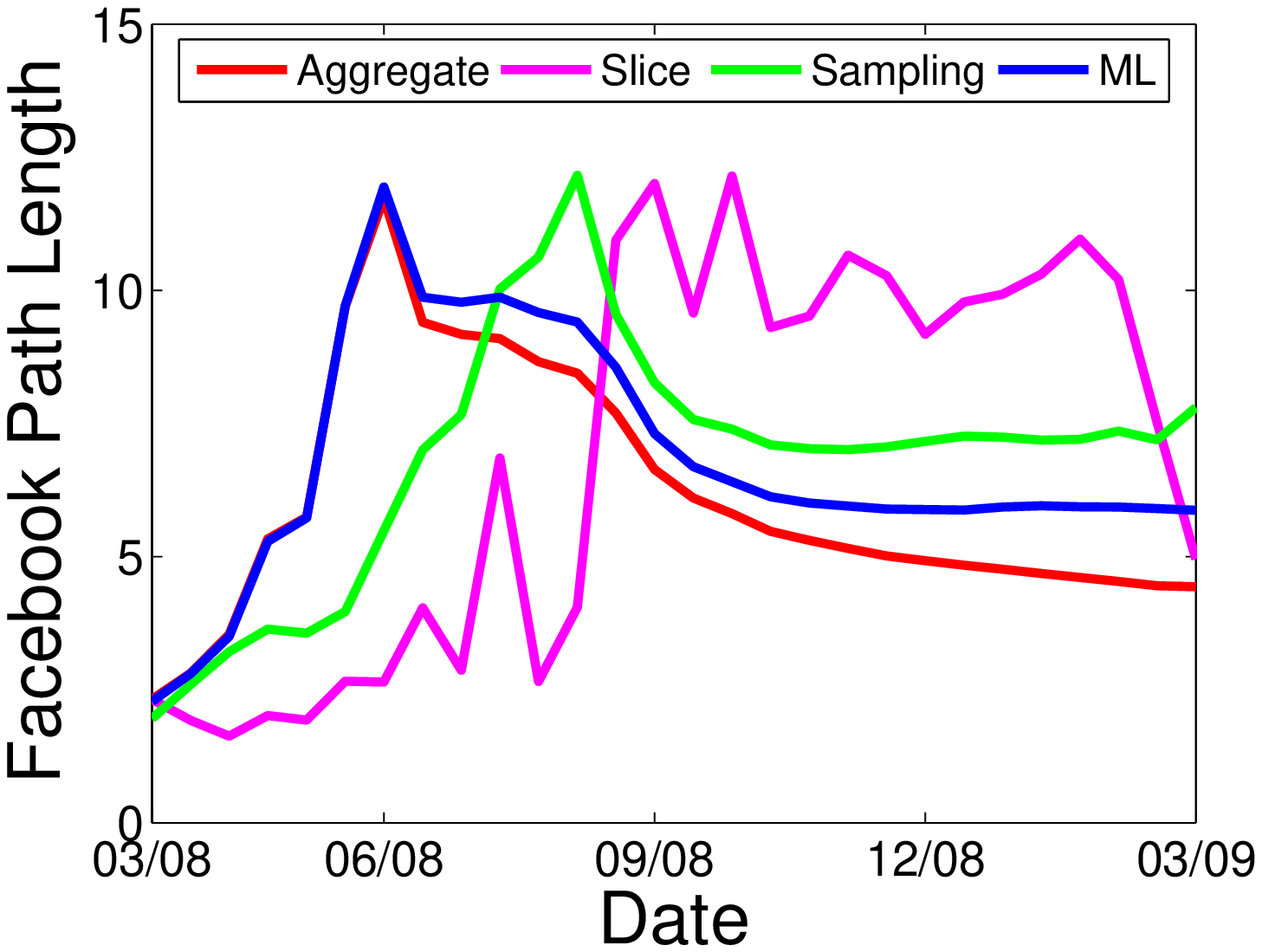}
\par\end{centering}}

\end{centering}
\vspace{-2.mm}
\caption{Average path lengths and clustering coefficients for Enron (a,b), DBLP (c,d) and Facebook (e,f).\label{fig:avg_CC}}
\vspace{-4.mm}

\end{figure}

In Figure~\ref{fig:avg_CC}, we report the average path lengths for the various measures: 
MLH paths, probabilistic shortest paths, the aggregate shortest paths and the slice shortest paths.  
Additionally, we report the average sampled clustering coefficient, the clustering coefficient approximation, 
and the aggregate and slice discrete clustering coefficients.  These are done for 
each of the three datasets through time, and we investigate changes in these global statistics 
to understand what, if any, changes occur with respect to the \emph{small world} network structure of the data 
\cite{clusteringcoefficient}.

In Figures \ref{fig:avg_CC}.a,c,e, we show the clustering coefficients for each
of the three datasets.  The aggregate graph significantly overestimates
the amount of \emph{current} clustering in the graph, while the slice
method is highly variable, especially for Enron.  In general, both
probabilistic measures are in between the two extremes, balancing the effects of recent data
and decreasing the long term effect of past information, with the MLH performing
similarly to the sampled clustering coefficient, and even better on
DBLP, where sampling undercuts the clustering (likely due to small sample
size).

Next, in Figures \ref{fig:avg_CC}.b,d,f, we examine the \emph{shrinking
diameter} of these small world networks \cite{Leskovec05graphsover}.  Here, the aggregate underestimates the
path length at a current point in time.  We can see that the most probable paths
closely follows the sampling results, with both lying between the
slice and aggregate measures while avoiding the variability of the slice method.

\section{Conclusions}

In this paper we investigated the problem of calculating centrality and clustering in
an uncertain network, and analyzed our methods using time evolving networks. We
demonstrated the limitation of using an aggregate graph representation to
capture uncertainty in the network structure due to changes over time, as well as the
limitation of using a slice-based representation due to its extreme variability.  We introduced 
sampling-based measures for average shortest path and betweenness
centrality, as well as measures based on the most probable
paths, which are more intuitive for capturing network flow.
We also outlined exact methods for the computation of most probable paths (and
by extention, most probable betweenness centrality), and incorporated the notion of
transmission probability.  
Additionally, we
developed a probabilistic clustering coefficient and gave a first order Taylor
expansion approximation for computation.

We provided empirical evidence on the Enron, DBLP, and Facebook datasets showing the sampling
and MLH's intuitive centrality rankings for the Enron employees and Facebook members, as well as the global 
properties for all three. The probabilistic centrality and clustering formulations are inherently
smoother than the measures computed from discretized time slices, however they can reason about {\em likely} change in graph structure due to changes over time, unlike the aggregate method, which includes all past information. We see
the MLH formulation is smoother than the sampling method, indicating
that the most probable paths through the graph may be more important to consider
than shortest paths.  Finally, we note that our experiments used a relatively simple estimate of relationship strength  for the edge probabilities in the network. In future work we will investigate alternative formulations of edge uncertainty.

{\small
\section{Acknowledgements}This material is based in part
upon work supported by the Intelligence Advanced Research Projects
Activity (IARPA) via Air Force Research Laboratory contract number
FA8650-10-C-7060.  The U.S. Government is authorized to reproduce and distribute
reprints for Government purposes notwithstanding any copyright
annotation thereon.  Disclaimer:  The views and conclusions contained
herein are those of the authors and should not be interpreted as
necessarily representing the official policies or endorsements, either
expressed or implied, of IARPA, AFRL or the U.S. Government.  Pfeiffer is
supported by a Purdue University Frederick N. Andrews Fellowship.
}
{\small
\bibliographystyle{aaai}
\bibliography{Pfeiffer_Neville}
}
\end{document}